\documentclass[12pt]{JHEP3}
\usepackage{amsthm,amsmath,amssymb,mathrsfs}

\DeclareMathAlphabet{\mathpzc}{OT1}{pzc}{m}{it}
\DeclareMathOperator{\Ima}{Im} \DeclareMathOperator{\rk}{rk}
\DeclareMathOperator{\Tr}{Tr} \DeclareMathOperator{\Div}{Div}
\DeclareMathOperator{\SL}{SL} \DeclareMathOperator{\GL}{GL}
\DeclareMathOperator{\Hei}{H} 

\title{Twisted Elliptic Genera of {\boldmath $\mathcal{N}=2$} SCFTs in
  Two Dimensions}

\author{
  Toshiya Kawai\\
  Research Institute for Mathematical Sciences,\\
  Kyoto University, Kyoto, Kyoto 606--8502, Japan }

\abstract{ The elliptic genera of two-dimensional $\mathcal{N}=2$
  superconformal field theories can  be twisted by the action of
  the integral Heisenberg group if their $U(1)$ charges are fractional.  
  The basic properties of the resulting twisted  elliptic genera and 
  the associated twisted Witten indices are
  investigated with due attention to their behaviors in
  orbifoldization. Our findings are illustrated by and applied to
  several concrete examples.  We give a better understanding of the
  duality phenomenon observed long before for certain Landau-Ginzburg
  models.  We revisit and prove an old conjecture of Witten which
  states that every $\mathsf{A}\mathsf{D}\mathsf{E}$ Landau-Ginzburg
  model and the corresponding minimal model share the same elliptic
  genus. Mathematically, we establish $\mathsf{A}\mathsf{D}\mathsf{E}$
  generalizations of the quintuple product identity.  }

\begin{document}

\newcommand{\CC}{{\mathbb{C}}} \newcommand{\GG}{{\mathbb{G}}}
\newcommand{\II}{{\mathbb{I}}} \newcommand{\JJ}{{\mathbb{J}}}
\newcommand{\RR}{\mathbb{R}} \newcommand{\QQ}{\mathbb{Q}}
\newcommand{\ZZ}{\mathbb{Z}} \newcommand{\HH}{\mathbb{H}}
\newcommand{\NN}{\mathbb{N}} \newcommand{\PP}{\mathbb{P}}
\newcommand{\abs}[1]{\lvert#1\rvert}
\newcommand{\Abs}[1]{\lVert#1\rVert}
\newcommand{\inv}[1]{\frac{1}{#1}}
\newcommand{\vev}[1]{\langle#1\rangle}
\newcommand{\abcd}{\begin{pmatrix}a&b\\c&d\end{pmatrix}}
\newcommand{\smallabcd}{\left(\begin{smallmatrix}a&b\\c&d\end{smallmatrix}\right)}
\newcommand{\bfe}[1]{{ \mathbf{e}\left[#1\right]}}
\newcommand{\GCD}[1]{\mathop{\rm gcd}(#1)}
\newcommand{\LCM}[1]{\mathop{\rm lcm}(#1)}

\newcommand{\bfL}{\Upsilon} \newcommand{\En}{\Delta}
\newcommand{\weight}{\omega} \newcommand{\twista}{r}
\newcommand{\twistb}{s} \newcommand{\wi}{\mathfrak{X}}

\theoremstyle{plain} \newtheorem{definition}[equation]{Definition}
\newtheorem{lem}[equation]{Lemma}
\newtheorem{prop}[equation]{Proposition}
\newtheorem{thm}[equation]{Theorem}
\newtheorem{conj}[equation]{Conjecture}
\newtheorem{cor}[equation]{Corollary}
\newtheorem{assumption}[equation]{Assumption}

\theoremstyle{remark} \newtheorem{rem}[equation]{Remark}
\newtheorem{ex}[equation]{Example}

\section{Introduction}\label{sec:intro}

Besides being basic tools for constructing superstring vacua,
$\mathcal{N}=2$ superconformal field theories (SCFTs) in two
dimensions are interesting subjects in their own rights and have long 
been studied by many authors.  Among other things, their elliptic
genera have drawn attention since they are often amenable to exact
computations and help us to explore possible relations among different
models.

In this paper, we discuss certain aspects of such genera of which much
notice has not been taken.  We consider the twisted versions of
elliptic genera with the twisting given by the action of the integral Heisenberg group.
 Actually, this sort of elliptic genera previously appeared in
the process of orbifoldization \cite{Kawai:1994uq}, but we will give a
systematic study of them here. The twisted elliptic genera we consider
are of interest only when the $U(1)$ charges of the theory  are
fractional  as in Landau-Ginzburg (LG) models or, say,
Kazama-Suzuki models.  If instead the $U(1)$ charges  are integral as 
is the case with any $\mathcal{N}=2$ sigma model
with a compact Calabi-Yau (CY) target space, the twisted elliptic
genera trivially reduce to the ordinary one.

We begin by briefly reviewing in \S \ref{sec:SF} the fundamental
notion of  the spectral flow for the $\mathcal{N}=2$ superconformal algebra
(SCA) \cite{Schwimmer:1987jk}.  In \S \ref{sec:EG}, we recast the
functional properties of elliptic genera (as proposed in
\cite{Kawai:1994uq}) into  suitable forms  so that the role of the Jacobi
group \cite{Eichler:1985kx} is  apparent.  We introduce the twisted
elliptic genera by applying the elements of the integral Heisenberg group to
the elliptic genera in \S\ref{sec:TEG}. 
For the twisted elliptic genera thus defined
we can also introduce  the twisted Witten indices
 exactly in the same way as the Witten index is associated with the ordinary
elliptic genus.  We explain how the twisted Witten indices are
related to the $\chi_y$-genus.  Then, we investigate in \S\ref{sec:Orb} the behaviors of  the twisted
elliptic genera and the twisted Witten indices under
orbifoldization.

When $\hat c<1$ (with $\hat c$ being (the one third of) the 
central charge), the information of the twisted Witten indices turns
out to be especially useful, since the elliptic genera of any two
theories with the same $\hat c<1$ satisfying the same functional
properties ought to be equal identically if their twisted Witten
indices coincide. This is an easy consequence of the lemma of Atkin
and Swinnerton-Dyer \cite{Atkin:1954lr} as will be explained in
\S\ref{sec:ASD}.

In the remaining sections, our general findings are illustrated by and
applied to several concrete examples. We treat general LG models and
their orbifolds in \S\ref{sec:LG} and explain in \S\ref{sec:duality}
how the insight gained in this work helps us to understand the
somewhat mysterious phenomenon observed in \cite{Kawai:1992kx}
concerning the dualities (or mirror symmetries) of certain LG models.
In \S\ref{sec:minimal} we discuss the $\mathcal{N}=2$ minimal models
and prove an old conjecture of Witten \cite{Witten:1994fj} which
states that every $\mathsf{A}\mathsf{D}\mathsf{E}$ LG model and the
corresponding minimal model share the same elliptic genus. We also
explain how this result can be interpreted mathematically as
$\mathsf{ADE}$ generalizations of the quintuple product identity
(QPI).  We point out that the identities for the $\mathsf{A}$-series
are essentially Bailey's generalizations of the QPI
\cite{Bailey:1953lr}.

\medskip {\noindent {\bf Convention}}: By a $\mathcal{N}=2$ SCFT we
always mean a $\mathcal{N}=(2,2)$ SCFT.

\medskip {\noindent {\bf Notations}}: We write $\bfe{x}$ for
$\exp(2\pi \sqrt{-1}x)$.  For a positive integer $n$  we set $ \zeta_n=\bfe{1/n}$
 and  write $\ZZ_n=\ZZ/n\ZZ$ for the integers modulo $n$.  We  denote by $\delta_{a,b}^{(n)}$ the Kronecker 
 delta symbol modulo $n$.  Namely,
$\delta_{a,b}^{(n)}$ is equal to $1$ if $a\equiv b \pmod{n}$ and
vanishes otherwise.

\section{The spectral flow of the {\boldmath $\mathcal{N}=2$}
  SCA} \label{sec:SF}
  
The $\mathcal{N}=2$ SCA $\mathfrak{A}_\varphi$ with the mode parameter
$\varphi\in \RR$ is a super Lie algebra over $\CC$ linearly spanned by
the bosonic center $\hat C$ whose eigenvalue is commonly written as
$\hat c$, the bosonic Virasoro generators $\{L _n\}_{n\in \ZZ}$, the
bosonic $U(1)$ current generators $\{J_n\}_{\in\ZZ}$, and the
fermionic supercurrent generators $\{G^\pm_k\}_{k\in \ZZ\mp \varphi}$.
The  non-vanishing super commutation relations are given by
\begin{equation}
  \begin{split}
    [L_m,L_n]&=(m-n)L_{m+n}+\frac{m^3-m}{4}\delta_{m+n,0}\,   \hat{C},\\
    [L_m,J_n]&=-nJ_{m+n},\\
    [J_m,J_n]&=m\delta_{m+n,0}\, \hat{C},\\
    [L_m,G^\pm_k]&=(\frac{m}{2}-k)G_{m+k}^\pm,\\
    [J_m,G_k^\pm]&=\pm G_{m+k}^\pm,\\
    [G^+_k,G^-_l]_+&=2L_{k+l}+(k-l)J_{k+l}+(k^2-\frac{1}{4})\delta_{k+l,0}\,\hat{C}.\\
   \end{split}
\end{equation}
The algebra $\mathfrak{A}_\varphi$ is called Ramond if $\varphi\in
\ZZ$ and Neveu-Schwarz if $\varphi\in \ZZ+1/2$.  The so-called twisted
$\mathcal{N}=2$ SCA is of no concern to us in this paper.

Actually, we have an isomorphism
$\mathfrak{A}_\varphi\cong \mathfrak{A}_{\varphi'}$ as super Lie algebras for any
$\varphi,\varphi'\in \RR$. 
If we set  $\eta=\varphi-\varphi'$, this
isomorphism is given by the {\em spectral flow\/}
$\sigma_\eta:\mathfrak{A}_{\varphi}\to \mathfrak{A}_{\varphi'}$ 
which one defines by 
\begin{align}
  \sigma_\eta(\hat{C})&=\hat{C},\\
  \sigma_\eta(L_n)&=L_n+\eta\, J_n+\frac{\eta^2}{2}\delta_{n,0}\, \hat{C},\\
  \sigma_\eta(J_n)&=J_n+\eta\, \delta_{n,0}\,\hat{C},\\
  \sigma_\eta(G^\pm_k)&=G^\pm_{k\pm\eta}\label{Gmoding}
\end{align}
together with $\CC$-linearity \cite{Schwimmer:1987jk, Lerche:1989cy}.
 The inverse of $\sigma_\eta$ is given
by $\sigma_{-\eta}$. 
If $\rho_{\varphi'}$ is a (highest weight) representation of
$\mathfrak{A}_{\varphi'}$, we also obtain a representation $\rho_\varphi$ of
$\mathfrak{A}_{\varphi}$ via
$\sigma_\eta$.  An important point to keep in mind  is that although we have
$\mathfrak{A}_\varphi\cong \mathfrak{A}_{\varphi'}$, the two
representations $\rho_\varphi$ and $\rho_{\varphi'}$ are in general
not equivalent.  This is so since the modes of $G^\pm_k$ shift as in
\eqref{Gmoding} modifying the highest weight conditions.

For instance, we will meet the following situation later. Consider the
Ramond $\mathcal{N}=2$ SCA $\mathfrak{A}_0$ and fix its representation
$\rho_0$ on some space $\mathcal{H}$.  The spectral flow
$\sigma_r:\mathfrak{A}_r\to \mathfrak{A}_0$ with $r\in \ZZ$ induces a
representation $\rho_r$ of another copy of the Ramond $\mathcal{N}=2$
SCA $\mathfrak{A}_r$ on $\mathcal{H}$. 
We can consider the subspaces
$\mathcal{V}_0\subset\mathcal{H}$ and
$\mathcal{V}_r\subset\mathcal{H}$ of the Ramond ground states for
$\rho_0$ and $\rho_r$ which are killed by the respective
$G_0^\pm$. However, $\mathcal{V}_0$ and $\mathcal{V}_r$ in general
have different dimensions since $\mathcal{V}_r$ corresponds to the
states annihilated by $G^\pm_{\pm r}$ in the representation $\rho_0$.
So, in particular, $\rho_0$ and $\rho_r$ will have different Witten
indices and they are different representations in general.  Obviously, a
similar statement applies  for $\rho_r$ and $\rho_{r'}$ with integers
$r\neq r'$. 

Before ending this brief section, we recall the most prominent application
of the spectral flow. The spectral flow
$\sigma_{1/2}:\mathfrak{A}_{1/2}\to \mathfrak{A}_{0}$ determines a
representation $\rho_{1/2}$ of the Neveu-Schwarz algebra
$\mathfrak{A}_{1/2}$.  As amply discussed in \cite{Lerche:1989cy},
$\mathcal{V}_0$ is bijectively mapped to the chiral ring $\mathcal{R}$
annihilated by $G^\pm_{\mp1/2}$ for $\rho_{1/2}$.

\section{Elliptic genera of {\boldmath $\mathcal{N}=2$}
  SCFTs} \label{sec:EG}
We begin with the setup of  a general $\mathcal{N}=2$ SCFT  to be studied in this work.

Let $\mathfrak{A}_{0}$ and $ \tilde{\mathfrak{A}}_{0}$ be respectively the left and right  Ramond $\mathcal{N}=2$ SCA of our $\mathcal{N}=2$ SCFT which is assumed to have  a non-negative rational  left=right  central charge $\hat c$.
(Below we  follow the convention  such that, given some object for the left-mover,
 the corresponding one for the right-mover is
written with a tilde.)  
Let  $\varrho_{0,0}$  be  the representation of  $(\mathfrak{A}_{0},\tilde{\mathfrak{A}}_{0})$  underlying the theory.    Let also  $\varrho_{r,r'}$ $(r,r'\in \ZZ)$  be the representation of
 $(\mathfrak{A}_{r},\tilde{\mathfrak{A}}_{r'})$ induced from $\varrho_{0,0}$ by the  spectral flow $(\sigma_r,\tilde \sigma_{r'})$.
We assume the existence of a positive integer $h$ such that
\begin{itemize}
\item  The central charge $\hat c$ can be expressed for some integers $\mathpzc{D}$ and $\delta$ as
\begin{equation}
\hat c=\mathpzc{D}-\frac{2\delta}{h}.
\end{equation} 
\item   The left and right $U(1)$ charges in $\varrho_{r,r'}$  belong to  $-\hat c/2+(1/h)\ZZ$.
\item There exist periodic isomorphisms $\varrho_{r+p,r'+p'}\cong \varrho_{r,r'}$ for any $p,p'\in h\ZZ$.
\end{itemize}
 We choose $h$ to be the smallest possible one in the
following.

Let $\HH$ be the upper half-plane $\{\tau\in\CC\mid\Ima \tau>0\}$. According to Witten, the elliptic genus is
defined by
\begin{equation}\label{def:eg}
  Z(\tau,z)=\Tr_{\varrho_{0,0}} (-1)^{F}(-1)^{\tilde F}q^{L _0-\hat c/8}
  \bar{q}^{\tilde L _0-\hat c/8}
  y^{J _0}, \quad (\tau,z)\in \HH\times \CC
\end{equation}
where  $(-1)^F$ and $(-1)^{\tilde F}$ are the usual left and right
fermion parity operators and we have set  $q=\bfe{\tau}$ and $y=\bfe{z}$. Due to supersymmetric
cancellations between bosonic and fermionic states above the ground
level for  the right-mover we may and will assume that the elliptic genus $Z(\tau,z)$ is a
holomorphic function on $\HH\times \CC$ having a Fourier expansion of
the form
\begin{equation}
  Z(\tau,z)=\sum_{n\ge 0}\mathfrak{Z}_n(z)q^n, \quad \mathfrak{Z}_n(z)\in \ZZ[y^{{1}/{h}},y^{-{1}/{h}}].
\end{equation}
In \cite{Kawai:1994uq}, we proposed that $Z(\tau,z)$ should in
addition satisfy
\begin{align}
  Z \left( \frac{a\tau +b}{ c\tau +d}, \frac{z}{c\tau +d} \right)&=
  \bfe{\frac{\hat c}{2}\frac{c z^2}{ c\tau +d}} Z(\tau, z)\,,\quad
  \begin{pmatrix} a&b\\ c&d \end{pmatrix} \in \SL_2({\ZZ}),\label{eg1}\\
  Z(\tau,z+\twista\tau+\twistb)&=(-1)^{\hat
    c(\twista+\twistb)}\bfe{-\frac{ \hat c}{2} (\twista^2\tau+2\twista
    z)}Z(\tau,z),\quad (\twista,\twistb)\in (h\ZZ)^2.\label{eg2}
\end{align}
These properties have been checked for many examples and will be
postulated in this work.
One can almost  derive \eqref{eg2}
from our assumption except the subtlety for the factor $(-1)^{\hat c r}$ which should  be understood form the behavior of $(-1)^F(-1)^{\tilde F}$ under the left spectral flow $\sigma_r$.
In fact,  if  the eigenvalues of $J_0-\tilde J_0$ are always integers so that  we can identify $(-1)^F(-1)^{\tilde F}$ with  $\bfe{\frac{1}{2}(J_0-\tilde J_0)}$, we see easily that $(-1)^{\hat c r}$ arises from $\sigma_r$.
This choice of $(-1)^F(-1)^{\tilde F}$ has been discussed, for instance, in \cite{Lerche:1989cy}.

The left $U(1)$ charge spectrum of  the chiral ring (or the  $(c,c)$-ring to be precise   \cite{Lerche:1989cy})   is captured by the
$\chi_y$-genus\footnote{ For a $\mathcal{N}=2$ sigma model with a
  compact Calabi-Yau target space $V$, we have $h=1$ and
  $\chi_y=\sum(-1)^{i+j}h^{i,j}(V)y^i$.  Thus our convention for the
  $\chi_y$-genus differs from the original one of Hirzebruch by the
  sign of $y$.}.  In view of the spectral flow $\sigma_{\pm 1/2}$, we
should have
\begin{equation}
  \mathfrak{Z}_0(z)=y^{-{\hat c}/{2}}\chi_y,\quad  \chi_y\in \ZZ[y^{1/h}].
\end{equation}
Since $Z(\tau,-z)=Z(\tau,z)$ by \eqref{eg1} we have the duality
 $\chi_{y^{-1}}=y^{-\hat c}\chi_y$.  As usual, the Witten index $\wi$ is
defined by
\begin{equation}
  \wi=Z(\tau,0)=\chi_y\vert_{y=1}.
\end{equation}

To develop our formalism below, it is expedient  to rephrase the properties
\eqref{eg1} and \eqref{eg2} in the language of the Jacobi group.  The
standard reference for this purpose is the monograph by Eichler and
Zagier \cite{Eichler:1985kx} (though we need a small extension of the
materials treated there.)  We first introduce the  symplectic product on
$\ZZ^2$ by
\begin{equation}
  v\wedge v'=vJ (v')^T=\det\begin{pmatrix}v\\v'
  \end{pmatrix}, \quad J=\begin{pmatrix}0&1\\-1&0\end{pmatrix}
\end{equation}
where $v,v'\in \ZZ^2$ are interpreted as row vectors.  We have
$v\wedge v'=-v'\wedge v$ and $vM\wedge v'M=v\wedge v'$ for any $M\in
\SL_2(\ZZ)\cong \mathrm{Sp}_2(\ZZ)$.  The integral Heisenberg group $\Hei(\ZZ)$
is the set $\ZZ^3=\{(v,\kappa)\mid v\in \ZZ^2,\kappa\in \ZZ\}$
equipped with the group multiplication
\begin{equation}
  (v,\kappa) \ast (v',\kappa')=(v+v',\kappa+\kappa'+v\wedge v').
\end{equation}
The modular (or symplectic) group $\SL_2(\ZZ)$ acts on $\Hei(\ZZ)$ by
$(v,\kappa)M=(vM,\kappa)$.  Then the Jacobi group
$G^J=\SL_2(\ZZ)\ltimes \Hei(\ZZ)$ is the group with the multiplication
law
\begin{equation} [M,X]\star[M',X']=[MM',XM' \ast X]
\end{equation}
where $M,M'\in \SL_2(\ZZ)$ and $X,X'\in \Hei(\ZZ)$.  For a positive
integer $n$, we also need to introduce
$\Hei(n\ZZ)=(n\ZZ)^3=\{(v,\kappa)\mid v\in (n\ZZ)^2,\kappa\in n\ZZ\}$
which is a subgroup of $\Hei(\ZZ)$.  Accordingly, we set
$G^J_n=\SL_2(\ZZ)\ltimes \Hei(n\ZZ)$.

Let $ \varepsilon:\Hei(\ZZ)\to \{1,-1\}$ be defined by
\begin{equation}
  \varepsilon:
  ((\twista,\twistb),\kappa)\mapsto (-1)^{\twista+\twistb+\twista\twistb+\kappa}.
\end{equation}
As one may easily confirm, this is a group homomorphism satisfying
$\varepsilon(XM)=\varepsilon(X)$ for any $M\in \SL_2(\ZZ)$.  For
$\ell\in \{0,1\}$, $m\in \QQ$, and $\phi(\tau,z)$ a holomorphic
function on $\HH\times \CC$, we set
\begin{equation}
  \phi\|_{\ell,m}X (\tau,z)= \varepsilon(X)^\ell\bfe{m(\twista^2\tau+2\twista z+\twista\twistb+\kappa)}\phi(\tau,z+\twista\tau+\twistb)
\end{equation}
where $X=((\twista,\twistb),\kappa)\in \Hei(\ZZ)$. This gives a
Schr{\" o}dinger type representation of $\Hei(\ZZ)$.  Namely, we have
\begin{equation}
  \phi\|_{\ell,m} X\|_{\ell,m} X'=\phi\|_{\ell,m} X \ast X',\quad X,X'\in \Hei(\ZZ).
\end{equation}
In the following it is useful to remember
\begin{equation}
  \phi\|_{\ell,m} X\|_{\ell,m} X'=\bfe{2m(v\wedge v')}\phi\|_{\ell,m} X'\|_{\ell,m} X \label{exchange}
\end{equation}
as well as
\begin{align}
  \phi\|_{\ell,m} (v,\kappa)&=\bfe{(\ell/2+m)\kappa}\,\phi\|_{\ell,m} (v,0),\label{hei1}\\
  \phi\|_{\ell,m}
  (v,\kappa)\|_{\ell,m}(v',\kappa')&=\bfe{(\ell/2+m)(v\wedge
    v')}\,\phi\|_{\ell,m} (v+v',\kappa+\kappa'). \label{hei2}
\end{align}

If we introduce
\begin{equation}
  \phi\vert_{k,m}M(\tau,z)
  =(c\tau+d)^{-k}\bfe{-\frac{mcz^2}{c\tau+d}}\phi\left(\frac{a\tau+b}{c\tau+d},\frac{z}{c\tau+d}\right)
\end{equation} 
for $k\in \ZZ$, $m\in \QQ$, and $M=\left(\begin{smallmatrix}
    a&b\\c&d\end{smallmatrix}\right) \in \SL_2(\ZZ)$, 
 we have
\begin{equation}
  \phi\vert_{k,m}M\vert_{k,m} M'=\phi\vert_{k,m}MM',\quad M,M'\in \SL_2(\ZZ)
\end{equation}
which yields a representation of $\SL_2(\ZZ)$.

Since a straightforward calculation shows that
\begin{equation}
  \phi\vert_{k,m}M\|_{\ell,m}XM=\phi\|_{\ell,m} X\vert_{k,m} M,  \label{com1}
\end{equation}
we obtain a representation of $G^J$ by setting
\begin{equation}
  \phi\vert_{k,\ell,m}[M,X]=\phi\vert_{k,m}M\|_{\ell,m} X.
\end{equation}

Now we can rephrase the properties \eqref{eg1} and \eqref{eg2} as
\begin{align}
  Z\vert_{0,\hat c/2}M&=Z,\quad  M\in \SL_2(\ZZ)\label{modular2},\\
  Z\|_{\tilde{\mathpzc{D}},\hat c/2} X&=Z,\quad X\in
  \Hei(h\ZZ)\label{period2}
\end{align}
or more concisely,
\begin{equation}
  Z\vert_{0,\tilde{\mathpzc{D}},\hat c/2} [M,X] =Z,\quad     [M,X]\in G^J_h
\end{equation}
where $\tilde{\mathpzc{D}}\in \{0,1\}$ is determined by 
$\mathpzc{D}\equiv \tilde{\mathpzc{D}} \pmod{2}$.

Let $\GG_m$ be the multiplicative group of non-zero complex numbers
and let $\CC\GG_m$ be the group ring of $\GG_m$ over $\CC$.  Given a
$\xi \in \GG_m$, the corresponding element in $\CC\GG_m$ is denoted by
$\vev{\xi}$.  For a monic polynomial $p(t)=\prod_{i=1}^k(t-\xi_i)$
with an indeterminate $t$ and $\xi_i\in \GG_m$, define $\Div(p(t))\in
\CC\GG_m$ by $\vev{\xi_1}+\cdots+\vev{\xi_k}$.  Set
$\Lambda_d=\Div(t^d-1)=\sum_{i\in \ZZ_d}\vev{\zeta_d^i}$ for a
positive integer $d$.  We will sometimes identify $\Lambda_1=\vev{1}$
with $1$.

Given a $\mathcal{N}=2$ SCFT  we  define
$\bfL\in \CC\GG_m$  as follows.
Suppose that $y^{\delta/h}\chi_y \in \ZZ[y^{1/h},y^{-1/h}]$ is written explicitly
 as $y^{\delta/h}\chi_y=\sum n_i y^{i/h}$.  Then we put $\bfL=\sum n_i\vev{\zeta_h^i}$.  This
can be uniquely expanded as\footnote{To see this, let $l_d$ be defined by
$l_d=u_{h/d}$ if $d\mid h$ and $l_d=0$ if otherwise.  Then, we should
have $n_i=\sum_{d\mid i}l_d$, which can be inverted as
$l_i=\sum_{d\mid i}\mu(d)n_{i/d}$ with the aid of the M{\" o}bius
function $\mu$.}
\begin{equation}\label{Gammaexp}
  \bfL=\sum_{d\mid h}u_d\Lambda_d
\end{equation}
for  some $u_d\in \ZZ$.
Notice that, for any $s\in \ZZ$,  we have
\begin{equation}
  y^{\delta/h}\chi_y\vert_{z=\twistb}=\sum_{d\mid \twistb,d\mid h}u_d d.  \label{chiy-ud}
\end{equation}
 In particular, the case  $s=0$ leads to
\begin{equation}
  \wi=\sum_{d\mid h}u_d d.
\end{equation}
Note also that we have $\Upsilon=u_1\Lambda_1=\wi\Lambda_1$ if  $h=1$.

\section{Twisted elliptic genera}\label{sec:TEG}

We define the twisted elliptic genus
$Z_v(\tau,z)$ for any $v\in \ZZ^2$
by the action of $(v,0)\in \Hei(\ZZ)$ on $Z(\tau,z)$, namely,
\begin{equation}
  Z_v:=Z\|_{\tilde{\mathpzc{D}},\hat c/2}(v,0).
\end{equation}
It is then straightforward to show that
\begin{align}
  Z_{v}\vert_{0,\hat c/2}M&=Z_{vM},\quad M\in \SL_2(\ZZ),\label{Zvmodular}\\
  Z_{v}\|_{\tilde{\mathpzc{D}},\hat c/2}X&=Z_{v}, \quad X\in
  \Hei(h\ZZ).\label{Zvperiod}
\end{align}
In fact, \eqref{Zvmodular} readily follows from \eqref{com1} and
\eqref{modular2} while \eqref{Zvperiod} follows from \eqref{exchange}
and \eqref{period2}.  We should also note that
\begin{equation}
  Z_{v}\|_{\tilde{\mathpzc{D}},\hat c/2}X=\zeta_h^{-(\kappa'+v\wedge v')\delta}Z_{v+v'},\quad X=(v',\kappa')\in \Hei(\ZZ)
  \label{Zvadd}
\end{equation}
which can be seen from \eqref{hei1} and \eqref{hei2}.  Combined with
\eqref{Zvperiod}, this implies the periodicity
\begin{equation}\label{Zvperiodicity}
  Z_v=Z_{v+v'},\quad  v' \in (h\ZZ)^2.
\end{equation}

Throughout this work we assume the following important property:
\begin{assumption}\label{tWgenus}
The twisted Witten index $\wi_v$ defined by $\wi_v:=Z_v(\tau,0)$ for any $v\in \ZZ^2$
 is independent of $\tau$.
\end{assumption} {\noindent In order} to explain why we believe this assumption
to be a reasonable one, we first write down the explicit expression of $Z_v$
as
\begin{equation}\label{extwist}
  Z_{(r,s)}(\tau,z)=(-1)^{{\mathpzc{D}}(\twista+\twistb+\twista\twistb)}
 \bfe{\frac{\hat c}{2}(\twista^2\tau+2\twista z+\twista\twistb)}Z(\tau,z+\twista\tau+\twistb).
\end{equation}
Then we recognize, after taking into account the properties of the 
spectral flow reviewed in \S \ref{sec:SF}, that $Z_{(r,0)}(\tau,z)$ is
essentially the {\em ordinary\/} elliptic genus with the trace in
\eqref{def:eg} being over  $\varrho_{r,0}$ instead of  $\varrho_{0,0}$.  Therefore $Z_{(r,0)}(\tau,0)$ must be a constant for any
$r\in \ZZ$.  Suppose that $(r,s)\in\ZZ^2$ is not equal to $(0,0)$.
Then there exist integers $m$ and $n$ such that
$m\twista+n\twistb=\GCD{\twista,\twistb}$.  If we put
\begin{equation}\label{Mchoice}
  M=\begin{pmatrix} \twista/\GCD{\twista,\twistb}&\twistb/\GCD{\twista,\twistb}\\
    - n&m
  \end{pmatrix}\in \SL_2(\ZZ)
\end{equation}
then \eqref{Zvmodular} means
\begin{equation}\label{Zvmodularsimp}
  Z_{(\GCD{r,s},0)}\vert_{0,\hat c/2}M=Z_{(r,s)}.
\end{equation}
Since the LHS of this is constant when $z=0$, we see that
$Z_{(r,s)}(\tau,0)$ must be constant as well. In any case, Assumption
\ref{tWgenus} can be directly  confirmed for all the examples we meet
below.

Once the assumption is accepted, we find rather stringent
constraints on $\wi_v$.  For instance, if $Z_v(\tau,z)=y^{-\hat
  c/2}(\chi_y)_v+O(q)$, then we should have
$\wi_v=(\chi_y)_v\vert_{y=1}$.  More striking is the fact that $\wi_v$
is essentially determined by $\bfL$ (hence by $\chi_y$).
 Observe first that
\begin{align}
  \wi_v&=\wi_{vM},\qquad  M\in \SL_2(\ZZ),\label{Xmodular}\\
  \wi_v&=\wi_{v+v'},\qquad v'\in (h\ZZ)^2. \label{Xperiod}
\end{align}
These follow from \eqref{Zvmodular} and \eqref{Zvperiodicity}.  W see
that $\wi_{(\twista,\twistb)}=\wi_{(\GCD{\twista,\twistb},0)}$ for any
$(r,s)\in \ZZ^2$.  Indeed, if $(r,s)=(0,0)$ this is trivially true
while if $(\twista,\twistb)\ne (0,0)$ this follows from \eqref
{Zvmodularsimp} by setting $z=0$.  That
$\wi_{(\twista,\twistb)}=\wi_{(\GCD{\twista,\twistb},0)}$ for any
$(r,s)\in \ZZ^2$ implies
$\wi_{(\twista,\twistb)}=\wi_{(\twistb,\twista)}$ for any $(r,s)\in
\ZZ^2$.  Consequently, $\wi_v$ actually has a larger symmetry than
\eqref{Xmodular}\footnote{Recall that generators of $\GL_2(\ZZ)$ can be taken as those of 
  $\SL_2(\ZZ)$ and $\left(\begin{smallmatrix}
      0&1\\1&0\end{smallmatrix}\right)$.}:
\begin{equation}
  \wi_v=\wi_{vM},\qquad  M\in \GL_2(\ZZ).\label{Xmodularext}
\end{equation}
Since
$\wi_{(\GCD{\twista,\twistb},0)}=\wi_{(0,\GCD{\twista,\twistb})}$, we
also have $\wi_{(\twista,\twistb)}=\wi_{(0,\GCD{\twista,\twistb})}$.
Furthermore, since we have
$\wi_{(0,\GCD{\twista,\twistb})}=\wi_{(h,\GCD{\twista,\twistb})}$ by
\eqref{Xperiod}, we may employ the same logic to see
$\wi_{(\twista,\twistb)}=\wi_{(0,\GCD{h,\twista,\twistb})}$.  Notice
then that
\begin{equation}
  \wi_{(0,\GCD{h,\twista,\twistb})}=(-1)^{{\mathpzc{D}}\GCD{h,\twista,\twistb}}y^{-\hat c/2}\chi_y\vert_{z= \GCD{h,\twista,\twistb}}=y^{\delta/h}\chi_y\vert_{z=\GCD{h,\twista,\twistb}}.
\end{equation}
By applying \eqref{chiy-ud} to this, we therefore obtain
\begin{equation}
  \wi_{(\twista,\twistb)}=\sum_{d\mid\GCD{ h,\twista,\twistb}}u_d d.
\end{equation}
In other words, we have shown, quite in parallel with
\eqref{Gammaexp},  that 
\begin{equation}\label{Xvexp}
  \wi_v=\sum_{d\mid h}u_d (\En_d)_v
\end{equation}
where we have introduced
\begin{equation}\label{deltadef}
  (\En_d)_v=\frac{1}{d}\sum_{v'\in (\ZZ_d)^2}\zeta_d^{v\wedge v'}
=\begin{cases} d&\text{if $d \mid \twista$ and $d \mid \twistb$ for $v=(\twista,\twistb)$,}\\
    0&\text{otherwise.} 
  \end{cases}
\end{equation}
Note that this in particular means $\wi_v\in \ZZ$.  Since we can
easily see that ${(\Delta_d)}_v={(\Delta_d)}_{vM}$ $(M\in
\SL_2(\ZZ))$,
$(\En_d)_{(\twista,\twistb)}=(\En_d)_{(\twistb,\twista)}$ and
$(\Delta_d)_v=(\Delta_d)_{v+v'}$ $(v'\in (h\ZZ)^2)$ for $d\mid h$, it
is obvious that \eqref{Xvexp} indeed satisfies \eqref{Xmodularext} and
\eqref{Xperiod}.  

It should be noted that  the expansion of the form \eqref{Xvexp} is  unique\footnote{This can be seen as follows.
Suppose that we have \eqref{Xvexp}. Then we have
$\wi_{(r,0)}=\sum_{d\mid r,d\mid h}u_d d=\sum_{d\mid r}u_d d$ where we have extended the definition of $u_d$  by setting $u_d=0$ if $d\nmid h$.
Then we have $u_d=\frac{1}{d}\sum_{r\mid d}\mu(r)\wi_{(d/r,0)}$ by the M{\" o}bius inversion formula.}.
Therefore, once either  \eqref{Gammaexp} or  \eqref{Xvexp} is known, the other has to follow. 
The expansion \eqref{Xvexp} turns out to be useful
when we investigate how $\wi_v$ (hence $\Upsilon$) behaves under the orbifoldization in
the next section.

\section{Orbifolds} \label{sec:Orb}

The orbifold elliptic genus is defined by averaging the twisted
elliptic genera \cite{Kawai:1994uq}\footnote{Of course, depending on
  the symmetry of the model, we may consider  more general types of
  orbifold theories and in fact such orbifolds are needed, say, when
  we consider mirror symmetry \`a la Greene-Plesser.  However, we
  restrict ourselves to the most fundamental one  in this
  paper.}:
\begin{equation}
  Z^{\rm orb}(\tau,z)=\inv{h}\sum_{v\in (\ZZ_h)^2} Z_v(\tau,z).
\end{equation}
Just like the ordinary elliptic genus, it obeys
\begin{align}
  Z^{\rm orb}\vert_{0,\hat c/2}M&=Z^{\rm orb},\quad M\in \SL_2(\ZZ),\label{Zorbmodular}\\
  Z^{\rm orb}\|_{\tilde{\mathpzc{D}},\hat c/2}X&=Z^{\rm orb}, \quad
  X\in \Hei(h^{\rm orb}\ZZ)\label{Zorbperiod}
\end{align}
or
\begin{equation}
  Z^{\rm orb}\vert_{0,\tilde{\mathpzc{D}},\hat c/2} [M,X] =Z^{\rm orb},\quad     [M,X]\in G^J_{h^{\rm orb}}
\end{equation}
where $h^{\rm orb}=h/\GCD{\delta,h}$.  
We see that \eqref{Zorbmodular}
follows from \eqref{Zvmodular}. To prove \eqref{Zorbperiod} we need to
use \eqref{Zvadd}.  As before, we introduce $\chi_y^{\rm orb}$ by
$Z^{\rm orb}(\tau,z)=y^{-{\hat c}/{2}}\chi_y^{\rm orb}+O(q)$.

By definition  $h^{\rm orb}$ is some divisor of $h$,  but  several extreme cases are noteworthy. If $\delta=0$ (or $\hat
c={\mathpzc{D}}$) we have $h^{\rm orb}=1$. In this case, our orbifold
procedure corresponds to Gepner's method for achieving charge
integrality.  Sometimes, the orbifold procedure can be regarded as the
mirror transformation.  Namely $Z^{\rm orb}$ may be interpreted as the
elliptic genus of the mirror theory within a sign.  For this to make
sense, it is necessary to have $h^{\rm orb}=h$ since the mirror transformation 
 must be involutive.  For instance, the relation  $h^{\rm
  orb}=h$ is obviously satisfied when $\delta=\pm 1$.  Later we will
discuss these extreme cases in more details taking examples from LG
models.

One can also consider the twisted version $Z^{\rm orb}_v$ which
apparently obeys the same functional equations as $Z_v$ except one has
to replace $h$ by $h^{\rm orb}$.  By using \eqref{Zvadd} we see that
\begin{equation}
  Z^{\rm orb}_v(\tau,z)=\inv{h}\sum_{v' \in (\ZZ_h)^2}\zeta_h^{(v\wedge v')\delta}Z_{v+v'}(\tau,z).
\end{equation}
From this it is obvious that $Z^{\rm orb}_v$ also satisfies Assumption
\ref{tWgenus}.  Thus, for $ \wi^{\rm orb}_v:=Z_v^{\rm orb }(\tau,0)$,
we have
\begin{equation}
  \wi^{\rm orb}_v=
  \inv{h}\sum_{v' \in (\ZZ_h)^2}\zeta_h^{(v\wedge v')\delta}\wi_{v+v'}. \label{Xorbv}
\end{equation}
Here we claim that
\begin{equation}
  \wi^{\rm orb}_v=\sum_{d\mid h}u_{h/d} \GCD{\delta,d}\Bigl(\En_{{d}/{\GCD{\delta,d}} }\Bigr)_v\label{Xorbvexp}
\end{equation}
which implies in particular $\wi^{\rm orb}_v\in \ZZ$ and
\begin{equation}
  \wi^{\rm orb}:=\wi^{\rm orb}_{(0,0)}=\sum_{d\mid h}u_{h/d}\,d.
\end{equation}
To see this, we expand $\wi_{v+v'}$ in \eqref{Xorbv} by using
$\wi_{v+v'}\allowbreak=\allowbreak \sum_{d\mid h}u_{h/d}
(\En_{h/d})_{v+v'}$.  Then it suffices to show that
\begin{equation}
  \frac{1}{h}\sum_{v'\in (\ZZ_h)^2}\zeta_h^{(v\wedge v')\delta}(\En_{h/d})_{v+v'}
  =\GCD{\delta,d}\Bigl(\En_{{d}/{\GCD{\delta,d}}}\Bigr)_v.\label{tmpex1}
\end{equation}
Since $(\En_{h/d})_{v+v'}$ is equal to $h/d$ if $v+v'=(h/d) v''$ for
$v''\in (\ZZ_{d})^2$ or $0$ otherwise, we have
\begin{equation}
  \frac{1}{h}\sum_{v'\in (\ZZ_h)^2}\zeta_h^{(v\wedge v')\delta}(\En_{h/d})_{v+v'}
=\frac{1}{d}\sum_{v''\in (\ZZ_{d})^2}\zeta_{d}^{(v\wedge v'')\delta}.
\end{equation}
Then \eqref{tmpex1} readily follows from \eqref{deltadef}.  

Since $h^{\rm
  orb}/(d/{\GCD{\delta,d}})=\LCM{\delta,h}/\LCM{\delta,d}$, we see that 
$d/{\GCD{\delta,d}}$ is a divisor of $h^{\rm orb}$  if  $d\mid h$.
 So
\eqref{Xorbvexp} can be recast in the form
\begin{equation}
  \wi_v^{\rm orb}=\sum_{d\mid h^{\rm orb}}u_d^{\rm orb} (\En_d)_v
\end{equation}
with an appropriate choice of integers $u^{\rm orb}_d$. For instance,
if $\delta=0$ for which $h^{\rm orb}=1$ we have $u_1^{\rm
  orb}=\wi^{\rm orb}=\sum_{d\mid h}u_{h/d}\,d$ while if $\delta=\pm 1$
for which $h^{\rm orb}=h$, we have $u^{\rm orb}_d=u_{h/d}$.

Previously, we defined $\bfL$ from the data of $\chi_y$ and studied
its relation to $\wi_v$. We may similarly associate $\bfL^{\rm orb}$
with $\chi_y^{\rm orb}$. Then $\bfL^{\rm orb}$ has similarly to be
related to $\wi_v^{\rm orb}$.  So we should have
\begin{equation}
  \bfL^{\rm orb}=\sum_{d\mid h^{\rm orb}}u_d^{\rm orb} \Lambda_d
\end{equation}
or
\begin{equation}
  \bfL^{\rm orb}=\sum_{d\mid h}u_{h/d} \GCD{\delta,d}\Lambda_{{d}/{\GCD{\delta,d}} }.\label{Gammaorbexp}
\end{equation}

\section{Some criteria for the equality of two elliptic
  genera}\label{sec:ASD}
 
Suppose that there are two $\mathcal{N}=2$ SCFTs  that are suspected to
be equivalent.  Denote respectively their elliptic genera, twisted
Witten indices and $\chi_y$-genera by $Z^{(i)}(\tau,z)$, $\wi_v^{(i)}$
and $\chi_y^{(i)}$ $(i=1,2)$.  Suppose furthermore that
$Z^{(1)}(\tau,z)$ and $Z^{(2)}(\tau,z)$ are known explicitly and are
confirmed to satisfy \eqref{eg1} and \eqref{eg2} with the same $\hat c$ and $h$.  
Then, it is natural to ask if $Z^{(1)}(\tau,z)=Z^{(2)}(\tau,z)$
identically.

With regard to this issue, we recall a frequently-used approach which is best suited when $\hat c
h^2\in \ZZ_{\ge 0}$ is relatively small.  If we put
$\phi^{(i)}(\tau,z')=Z^{(i)}(\tau,hz')$ then $\phi^{(i)}(\tau,z')$ is
a weak Jacobi form of weight zero and index $\hat c h^2/2$ (possibly
with multipliers).  The structure theorem over $\CC$ on the ring of
weak Jacobi forms with even weights and integral indices was 
proved in \cite[Theorem 9.3]{Eichler:1985kx}.  It is not difficult to
extend this theorem for the case with indices taking their values in
$\ZZ/2$.  To apply this sort of structure theorem to the present
problem we only have to check whether $\phi^{(1)}(\tau,z')$ and
$\phi^{(2)}(\tau,z')$ have the same $q$-expansions up to a certain
order.  However, the disadvantage of this approach is that as $\hat c
h^2$ becomes large, the necessary order for the $q$-expansions also
becomes large and we may not be able to carry out the test in
practice. Notice also that the approach is quite general but rather
weak since it does not utilize the properties of the elliptic genera
imposed by Assumption \ref {tWgenus}.

Here we explain another approach which also suffers from its own
limitation but is especially powerful when combined with Assumption
\ref{tWgenus} and applied to the cases with $\hat c<1$.  It employs
the fundamental lemma of Atkin and Swinnerton-Dyer
\cite{Atkin:1954lr}:
\begin{lem} Fix a constant    $q$ with $0<\abs{q}<1$ and set
  $\mathcal{A}=\{x\in \CC \mid \abs{q}<\abs{x}\le 1\}$.
  Suppose that a meromorphic function $\mathcal{F}(x)$ on $\CC^\times$
  satisfies $\mathcal{F}(qx)=K x^{-n}\mathcal{F}(x)$ for some integer
  $n$ and some constant $K$.  Then either $\mathcal{F}(x)$ has
  exactly $n$ more zeros than poles in $\mathcal{A}$ or
  $\mathcal{F}(x)$ vanishes identically.
\end{lem}

\begin{prop}\emph{(\emph{cf.}~\cite[Theorem 1.2]{Eichler:1985kx}.)}
 The elliptic genus $Z(\tau,z)$ of a $\mathcal{N}=2$ SCFT either has exactly $\hat c h^2$ zeros in
$\mathcal{P}:=\{s+t\tau \mid (s,t)\in[0,h)^2\}$ as a function of $z$ or
vanishes  identically.
\end{prop}
\begin{proof}
Set $\mathcal{F}_Z(x)=Z(\tau,z)$ with $y=x^h$.  Obviously, we have
\begin{equation}
  \mathcal{F}_Z(qx)=(-1)^{\hat c h}q^{-\hat ch^2/2}x^{-\hat c h^2}\mathcal{F}_Z(x).
\end{equation}
Since $\mathcal{F}_Z(x)$ is holomorphic in $\mathcal{A}$ (by our
assumption on $Z(\tau,z)$), the lemma implies that it has exactly
$\hat c h^2$ zeros in $\mathcal{A}$ or is identically zero. 
\end{proof}

\begin{prop}\label{egcriteria}
Given  the elliptic genera  $Z^{(1)}(\tau,z)$ and $Z^{(2)}(\tau,z)$ of  two $\mathcal{N}=2$ SCFTs  
 satisfying \eqref{eg1} and \eqref{eg2}
with the same $\hat c$ and $h$,  we have the following criteria for their equality:
\begin{enumerate}
\item[{\rm (i)}]  If $Z^{(1)}(\tau,z)=Z^{(2)}(\tau,z)$ at more than
    $\hat c h^2$ values of  $z$ in $\mathcal{P}$, then
    $Z^{(1)}(\tau,z)=Z^{(2)}(\tau,z)$ identically. 
\item[{\rm (ii)}] If $\hat c<1$ and $\wi_v^{(1)}=\wi_v^{(2)}$ for all
    $v\in (\ZZ_h)^2$, then $Z^{(1)}(\tau,z)=Z^{(2)}(\tau,z)$
    identically.
\item[{\rm (iii)}] If $\hat c<1$ and $\chi_y^{(1)}=\chi_y^{(2)}$, then
    $Z^{(1)}(\tau,z)=Z^{(2)}(\tau,z)$ identically.
\end{enumerate}
\end{prop}
\begin{proof} 
Put 
$\mathcal{F}(x)=\mathcal{F}_{Z^{(1)}}(x)-\mathcal{F}_{Z^{(2)}}(x)$
where  $\mathcal{F}_{Z^{(i)}}(x)$ $(i=1,2)$ is as in the proof of the previous proposition.
Since $\mathcal{F}(x)$ is holomorphic in $\mathcal{A}$,  the statement  (i) follows from the lemma.
Suppose  that $\hat c<1$. In view of \eqref{extwist},
$\wi_v^{(1)}=\wi_v^{(2)}$ $(v\in (\ZZ_h)^2)$ implies
$Z^{(1)}(\tau,z)=Z^{(2)}(\tau,z)$ at $h^2(>\hat c h^2)$ integral
points $\{r\tau+s\mid r,s=0,1,\dots,h-1\}\subset \mathcal{P}$. Hence we obtain
(ii) from (i).  The relation we studied between $\wi_v$ and $\chi_y$  implies that
(iii) is a consequence of (ii) (so long as we can check our Assumption
\ref{tWgenus}).
 \end{proof}
 
\section{LG models and their orbifolds} \label{sec:LG}

We next turn to LG models and their orbifolds to illustrate our general formalism.
We start by  fixing  some notation. We introduce the normalized Jacobi theta function $\vartheta$ on $
\HH\times\CC$ by
\begin{equation}
  \vartheta(\tau,u)=(x^{-1/2}-x^{1/2})\prod_{n=1}^\infty \frac{(1-q^n x)(1-q^n x^{-1})}{(1-q^n)^2}
\end{equation}
where $x=\bfe{u}$. This has no poles but simple zeros at
$\twista\tau+\twistb$ $(\twista,\twistb\in\ZZ)$ as a function of $u$
and obeys
\begin{align}
  \vartheta\vert_{-1,\frac{1}{2}}M&=\vartheta,\quad  M\in \SL_2(\ZZ),\\
  \vartheta\|_{1,\frac{1}{2}} X&=\vartheta,\quad X\in \Hei(\ZZ).
\end{align}

Let $A=\CC[x_1,\dots,x_n]$ be the ring of polynomials. Consider a
weighted homogeneous polynomial $f\in A$ satisfying $f(0,\dots,0)=0$
and
\begin{equation}\label{scaling}
  f(t^{\weight_{1}}x_1,\dots,t^{\weight_n}x_n)=tf(x_1,\dots,x_n),\quad t\in \CC^\times
\end{equation}
where the weight $\weight_i$  are positive rational numbers. We
express $\weight_i$ as an irreducible fraction $b_i/a_i$ so that
$a_i,b_i>0$ and $\GCD{a_i,b_i}=1$.  Set $N= \LCM{a_1,\dots,a_n}$.  Let
$I_f\subset A$ be the ideal generated by the partial derivatives
$\partial_1 f,\dots,\partial_nf$.  We assume that $(0,\dots,0)\in
\CC^n$ is the unique common zero of $\partial_1 f,\dots,\partial_nf$.

The elliptic genus of the $\mathcal{N}=2$ LG model with superpotential
$f$ is given by
\begin{equation}\label{LGeg}
  Z(\tau,z)=\prod_{i=1}^n\frac{\vartheta(\tau,(1-\weight_i)z)}{\vartheta(\tau,\weight_i z)}.\end{equation}
This is a straightforward generalization of Witten's result \cite{Witten:1994fj} for the simplest case 
$f(x_1)=x_1^{h}$ with $\weight_1=1/h$.
Behind this expression is
the so-called $bc\beta\gamma$ realization of $\mathcal{N}=2$ SCA
\cite{Fre:1991yq,Fre:1992kx,Witten:1994fj,Kawai:1994uq}.
It is easy to see that \eqref{LGeg} obeys the transformation laws \eqref{eg1} and \eqref{eg2} with
$\hat c=\sum_{i=1}^n(1-2\weight_i)$ and an appropriate choice of  $h$. The actual value of $h$ depends on
the detail of $f$ but we should have $h \mid N$. We may then take $\mathpzc{D}$ and $\delta$ so that
$\mathpzc{D}\equiv n \pmod{2}$ and $\delta/h- \sum\omega_i \in \ZZ$.
One can also readily confirm   that  \eqref{LGeg} has only removable singularities and is in fact holomorphic.
It has $\sum(1-\omega_i)^2h^2-\sum\omega_i^2h^2=\hat c h^2$ zeros in $\mathcal{P}$ as expected.

The chiral ring $\mathcal{R}$ is nothing but the Jacobi ring
$J_f:=A/I_f$ which has the grading $J_f=\oplus_\nu(J_f)_\nu$ induced
from \eqref{scaling}.  The $\chi_y$-genus is simply the Poincar{\' e}
polynomial of $J_f$:
\begin{equation}\label{LGchiy}
  \chi_y=\sum_\nu \dim (J_f)_\nu\, y^\nu=\prod_{i=1}^n\frac{1-y^{1-\weight_i}}{1-y^{\weight_i}}\,.
\end{equation}
The Witten index is equal to the Milnor number:
$\wi=\prod_{i=1}^n\left(\frac{1}{\weight_i}-1\right)$.  On the other
hand the space of Ramond ground states $\mathcal{V}_{0}$ should rather
be identified with $\Omega_A^n/\mathrm{d}f\wedge\Omega_A^{n-1}$ where
$\Omega_A^\bullet$ is the complex of K{\" a}hler differentials.  The
$U(1)$ charge assignments for $x_i$ and $\mathrm{d}x_i$ are
respectively $\weight_i$ and $\weight_i-1/2$.  Thus
$\mathrm{d}x_1\wedge\dots\wedge \mathrm{d}x_n$ has the $U(1)$ charge
$\sum_i(\weight_i-1/2)=-{\hat c}/{2}$. This explains the necessary
shift of the $U(1)$ charge under the spectral flow.

Before turning to the explicit formula of $\bfL$, let us observe that,
in the present case, we have $\bfL=\Div(\Phi(t))$ where $\Phi(t)$ is
the characteristic polynomial of the Milnor monodromy
\cite{Milnor:1968lr}.  To see this, set $ l(\nu)=\nu+\sum\omega_i-1 $
for any $\nu\in \QQ$.  Let $\{\phi_i\}$ be a basis of $J_f$ as a
$\CC$-vector space such that $\phi_i\in (J_f)_{\nu_i}$.  Since the
$(n-1)$-th cohomology group of the Milnor fiber is spanned by the
Gel'fand-Leray forms $\{\phi_i\, \mathrm{d}x_1\wedge\cdots\wedge
\mathrm{d}x_n/\mathrm{d}f\}$, the (cohomological) Milnor monodromy has
the spectrum $\{\bfe{l(\nu_i)}\}$
\cite{Steenbrink:1977fk,Arnold:1988qy}.  Consequently, we have
\begin{equation}
  \Div(\Phi(t))=\sum_\nu \dim(J_f)_\nu\vev{\bfe{l(\nu)}}.
\end{equation}
Then this is easily seen to coincide with our definition of $\bfL$:
\begin{equation}
  \bfL=\vev{\bfe{\delta/h}}
  \sum_\nu \dim(J_f)_\nu\vev{\bfe{\nu}}.
\end{equation}

In practice, we can calculate $\bfL$ from $\chi_y$ following the trick
of Orlik and Solomon \cite{Orlik:1977lr}.  Clearly, we should have
\begin{equation}
  \bfL=\vev{\bfe{\delta/h}}
  \lim_{y\to 1}\prod_{i=1}^n\frac{1-\vev{\bfe{-\weight_i}}y^{1-\weight_i}}{1-\vev{\bfe{\weight_i}}y^{\weight_i}}
  =\lim_{y\to 1}\prod_{i=1}^n\frac{\vev{\bfe{\weight_i}}-y^{1-\weight_i}}{1-\vev{\bfe{\weight_i}}y^{\weight_i}}.
\end{equation}
The evaluation of the limit is rather subtle, but if we notice that
\begin{equation}\label{factor}
  \begin{split}
    \frac{\vev{\bfe{\weight_i}}-y^{1-\weight_i}}{1-\vev{\bfe{\weight_i}}y^{\weight_i}}&=y^{-\weight_i}\left(\frac{1-y}{1-\vev{\bfe{\weight_i}}y^{\weight_i} }-1\right)\\
    &=y^{-\weight_i}\left(\frac{1-y}{1-y^{b_i}}\sum_{k=0}^{a_i-1}\vev{\bfe{k\weight_i}}y^{k\weight_i}
      -1\right)
  \end{split}
\end{equation}
one finds that
\begin{equation}\label{MO}
  \bfL=\prod_{i=1}^n\left(\inv{b_i}\Lambda_{a_i}-1\right).
\end{equation}
This is precisely the formula of Milnor and Orlik \cite{Milnor:1970fk}
for $\Div(\Phi(t))$ in accordance with our anticipation.  By
repeatedly employing a useful formula \cite{Milnor:1970fk}
\begin{equation}\label{mformula}
  \Lambda_a\Lambda_b=\GCD{a,b}\Lambda_{\LCM{a,b}}
\end{equation}
one obtains the expansion \eqref{Gammaexp} with
\begin{equation}
  u_1=(-1)^n\quad\text{and}\quad
  u_d= \frac{1}{d}\sum_{\LCM{a_{i_1},\ldots,a_{i_\ell}}=d}\frac{(-1)^{n-\ell}}{\weight_{i_1}\cdots\weight_{i_\ell}}\quad \text{for $d>1$.}
  \label{udLG}
\end{equation}
Note that $u_d$ given by \eqref{udLG} has to vanish for a divisor $d$
of $N$ greater than $h$.
\begin{rem}
  The work of A'Campo \cite{ACampo:1975fj} enables us to interpret the
  expansion \eqref{Gammaexp} in terms of the data associated with the
  resolution of the singularity.  On the other hand, from the
  perspective of $\mathcal{N}=2$ SCFTs, \eqref{Gammaexp} is related to
  the Coulomb gas decomposition in the case of $\mathsf{ADE}$ minimal models (cf.~\eqref{Omegadecomp}).  We pointed out this parallelism
  before  \cite{Kawai:1992kx} but a satisfactory understanding is
  still lacking.
\end{rem}

Since one can check Assumption \eqref{tWgenus}, our previous argument
implies that $\wi_v$ must have the expansion \eqref{Xvexp} with the
$u_d$ given by \eqref{udLG}.  This can also be seen in a direct
computation.  Indeed, by substituting \eqref{LGeg} into
\eqref{extwist} it is straightforward to show
\begin{equation}
  \wi_v=(-1)^n\prod_{i\in \mathcal{S}_v} \left(1-\inv{\weight_i}\right)
\end{equation}
where $\mathcal{S}_v$ is defined by $\mathcal{S}_{(r,s)}=\{i\in\{1,\dots,n\} \mid \text{$\twista\weight_i
  \in\ZZ$ and $\twistb\weight_i \in \ZZ$}\}$.  Then we can easily expand this expression to
obtain the expected result.  It should be noted that we have the (quasi) periodicity:
\begin{equation}
  \wi_v[f+x_{n+1}^2]=\varepsilon((v,0))\wi_v[f],\quad 
  \wi_v[f+x_{n+1}^2+x_{n+2}^2]=\wi_v[f].
\end{equation}

As we have already seen,  we are bound to have \eqref{Xorbvexp} and
\eqref{Gammaorbexp} for the orbifold theory.  Nevertheless,
here we prove \eqref{Gammaorbexp} directly for the sake of
completeness.  By explicitly calculating the expansion
$Z_v(\tau,z)=y^{-\hat c/2}(\chi_y)_v+O(q)$ one finds
\cite{Kawai:1995nx} that
\begin{equation}
  (\chi_y)_v=\prod_{i\not \in S_\twista}-y^{\inv{2}(1-2\weight_i)-(\!(\weight_i\twista)\!)}
  \prod_{i\in S_\twista}
  \frac{\bfe{\weight_i\twistb}-y^{1-\weight_i}}{1-\bfe{\weight_i\twistb}y^{\weight_i}}
\end{equation}
where we have set $S_\twista=\{i\mid \weight_i\twista\in \ZZ\}$ and $(\!(\weight
)\!)=\weight-\lfloor \weight\rfloor-1/2$ with $\lfloor \omega\rfloor=\max\{n\in \ZZ\mid n\le \omega\}$.
 Therefore,
\begin{equation}
  \chi_y^{\rm orb}=\frac{1}{h}\sum_{v \in (\ZZ_h)^2}(\chi_y)_v.
\end{equation}
This reproduces Vafa's  expression for  $\chi_y^{\rm orb}$  
\cite{Vafa:1989gd} originally found by a judicious physical argument but without
recourse to the elliptic genus.  As before, one can calculate
$\bfL^{\rm orb}$ from this expression of $\chi_y^{\rm orb}$ as
\begin{equation}\label{Phiorb}
  \bfL^{\rm orb}=\frac{1}{h}\sum_{v \in (\ZZ_h)^2}\bfL_v
\end{equation}
where
\begin{equation}
  \bfL_{(r,s)}=\vev{\zeta^{-\delta\twista}_h}\lim_{y\to 1}\prod_{i\not\in S_\twista }-y^{\inv{2}(1-2\weight_i)-(\!(\weight_i\twista)\!)}
  \prod_{i\in S_\twista}
  \frac{\bfe{\weight_i\twistb}\vev{\bfe{\weight_i}}-y^{1-\weight_i}}{1-\bfe{\weight_i\twistb}\vev{\bfe{\weight_i}}y^{\weight_i}}.
\end{equation}
To evaluate the limit we use as before
\begin{equation}
  \begin{split}
    \frac{\bfe{\weight_i\twistb}\vev{\bfe{\weight_i}}-y^{1-\weight_i}}{1-\bfe{\weight_i\twistb}\vev{\bfe{\weight_i}}y^{\weight_i}}\
    &=y^{-\weight_i}\left(\frac{1-y}{1-\bfe{\weight_i\twistb}\vev{\bfe{\weight_i}}y^{\weight_i} }-1\right)\\
    &=y^{-\weight_i}\left(\frac{1-y}{1-y^{b_i}}\sum_{k=0}^{a_i-1}\bfe{\weight_i\twistb
        k}\vev{\bfe{\weight_ik}}y^{\weight_ik} -1\right).
  \end{split}
\end{equation}
Then it turns out that
\begin{equation}
  \bfL_{(r,s)}=(-1)^n\vev{\zeta^{-\delta\twista}_h} \prod_{i\in S_\twista}\left(1-\inv{b_i}
    \sum_{k=0}^{a_i-1}\bfe{\weight_i\twistb k}\vev{\bfe{\weight_i k}} \right).
\end{equation}
By plugging this into \eqref{Phiorb} and expanding the resulting
expression, we obtain
\begin{equation}
  \begin{split}
    \bfL^{\rm orb}&=\frac{(-1)^n}{h}\sum_{(\twista,\twistb)\in (\ZZ_h)^2}\vev{\zeta^{-\delta\twista}_h}\sum_\ell\sum_{i_1,\ldots,i_\ell\in S_\twista}\frac{(-1)^\ell}{b_{i_1}\cdots b_{i_\ell}}\\
    &\times
    \sum_{k_{1}=0}^{a_{i_1}-1}\cdots\sum_{k_{\ell}=0}^{a_{i_\ell}-1}
    \mathbf{e}[(\sum_{j=1}^\ell \weight_{i_j}k_j)\twistb] \vev{
      \mathbf{e}[\sum_{j=1}^\ell \weight_{i_j}k_j]}.
  \end{split}
\end{equation}
We then perform the sum over $\twistb$ to find
\begin{equation}
  \bfL^{\rm orb}=\sum_{d\mid h}u_d\sum_{\substack{\twista\in \ZZ_h\\d\mid \twista}}\vev{\zeta^{-\delta \twista}_h}=\sum_{d\mid h}u_{h/d}\sum_{\substack{\twista\in \ZZ_h\\ (h/d)\mid \twista}}\vev{\zeta^{-\delta \twista}_h}
\end{equation}
where the $u_d$ are as given by \eqref{udLG}.  We can confirm without
difficulty that
\begin{equation}
  \sum_{\substack{\twista\in \ZZ_h\\ (h/d)\mid \twista}}\vev{\zeta^{-\delta \twista}_h}=\GCD{\delta,d}\Lambda_{{d}/{\GCD{\delta,d}}}\quad\text{if  $d\mid h$}.
\end{equation}
So we have completed the proof of \eqref{Gammaorbexp}.

\section{Some LG examples with {\boldmath $\delta=0,\pm 1$} and their
  dualities} \label{sec:duality}

The orbifold theories of some $\mathcal{N}=2$ SCFTs have chances to be
equivalent or at least closely related to some other known
$\mathcal{N}=2$ SCFTs. In order to appreciate  what we
have learnt so far, we give below some examples of LG models with
$\delta=0,\pm 1$ for which $(-1)^{\mathpzc{D}}Z^{\rm orb}(\tau,z)$ is
(expected to be) equal to the elliptic genus of another familiar
$\mathcal{N}=2$ SCFT\footnote{The sign here is purely conventional and
  one may prefer to include this in the definition of $Z^{\rm orb}$.}.

First, consider a LG model with $\delta=0$. Via the well-recognized  CY/LG
correspondence we expect that $(-1)^{\mathpzc{D}}Z^{\rm orb}(\tau,z)$
is equal to the elliptic genus of the $\mathcal{N}=2$ sigma model with
its target CY $\mathpzc{D}$-fold being (a resolution of) the
hypersurface $f=0$ in the pertinent weighted projective space. As
mentioned in \S \ref{sec:ASD}, this can be confirmed, if $\mathpzc{D}$
is not so large, by computing and comparing several terms in the
$q$-expansions of $(-1)^{\mathpzc{D}}Z^{\rm orb}(\tau,z)$ and the
elliptic genus of the $\mathcal{N}=2$ sigma model.  Since $h^{\rm
  orb}=1$, we have $\Upsilon^{\rm orb}=\wi^{\rm orb}\Lambda_1$ with
$\wi^{\rm orb}=\sum_{d\mid h}u_{h/d}d$ and $(-1)^{\mathpzc{D}}\wi^{\rm
  orb}$ has to be the Euler characteristic of the CY
$\mathpzc{D}$-fold.  For instance, consider the Fermat type
superpotential $f=x_1^n+\cdots+x_n^n$ for which we can take
$\mathpzc{D}=n-2$ and $\delta=0$.  Then,
\begin{equation}
  \Upsilon=\prod_{i=1}^n(\Lambda_n-1)^n=\frac{(n-1)^n-(-1)^n}{n}\Lambda_n+(-1)^n\Lambda_1.
\end{equation}
Hence we obtain $\Upsilon^{\rm
  orb}=(-1)^n\left(n+\frac{(1-n)^n-1}{n}\right) \Lambda_1$.  This
predicts $\chi(V_n)=n+\frac{(1-n)^n-1}{n}$ for the CY manifold
$V_n:=\{f=0\} \subset \PP^{n-1}$, which is indeed true.

Next we discuss several (well-known) LG models with $\delta=\pm 1$ and
explain how the duality or mirror phenomena observed long before in
\cite{Kawai:1992kx} can be understood from the vantage viewpoint of
the present work.  The LG models we consider are in three variables
and are distinguished by their types $\mathsf{T}$. We express their
weights $\weight_i$ as $d_i/h$ for $i=1,2,3$.

We say that $\mathsf{T}$ is {\it simple\/} or {\em of $\mathsf{ADE}$
  type\/} if it is given by the following data:
\begin{equation}\label{ADEtypes}
  \left.\begin{array}{llll} \mathsf{T}& (h,d_1,d_2,d_3)\qquad&f \\[1mm]\hline \mathsf{A}_{h-1}&(h,1,\frac{h}{2},\frac{h}{2}) & x_1^h+x_2^2+x_3^2\\[1mm] 
      \mathsf{D}_{\frac{h}{2}+1}\quad&(h, 2,\frac{h}{2}-1,\frac{h}{2})&x_1^{\frac{h}{2}}+x_1x_2^2+x_3^2\\[1mm]  
      \mathsf{E}_6&(12,3,4,6) &x_1^4+x_2^3+x_3^2\\[1mm]
      \mathsf{E}_7& (18,4,6,9)&x_1^3x_2+x_2^3+x_3^2\\[1mm]
      \mathsf{E}_8 & (30,6,10,15)&x_1^5+x_2^3+x_3^2\end{array}\right.
\end{equation}
where we assume that $h\ge 2$ for $\mathsf{A}_{h-1}$ and $h\ge 6$,
$h\in 2\ZZ$ for $\mathsf{D}_{\frac{h}{2}+1}$ and $h$ is the Coxeter
number of the Lie algebra corresponding to $\mathsf{T}$.  We have
$\hat c=1-2/h$ so that we may take $\mathpzc{D}=1$ and $\delta=1$. We
also see that $\wi=\rk$ where $\rk$ is the rank of the Lie algebra and that
$\chi_y=\sum_{i=1}^{\rk}t^{m_i-1}$ with $y=t^{1/h}$ where  the $m_i$ are
the Coxeter exponents:
\begin{equation}
  \left.\begin{array}{lll} 
      \mathsf{T} & m_i  \\\hline
      \mathsf{A}_{h-1}&1,2,\dots,h-1\\
      \mathsf{D}_{\frac{h}{2}+1}\quad&1,3,\dots,h-3,h-1,\frac{h}{2}\\
      \mathsf{E}_6& 1,4,5,7,8,11\\
      \mathsf{E}_7&1,5,7,9,11,13,17\\
      \mathsf{E}_8 &1,7,11,13,17,19,23,29\end{array}\right. 
\end{equation}
With the help of \eqref{MO} it is easy to see that
\begin{equation}\label{ADEexp}
  \left.\begin{array}{ll} \mathsf{T} & \bfL  \\\hline \mathsf{A}_{h-1}&\Lambda_h-\Lambda_{1}   \\  \mathsf{D}_{\frac{h}{2}+1}\quad&  \Lambda_{h}-\Lambda_{\frac{h}{2}}+\Lambda_{2}-\Lambda_1 \\  \mathsf{E}_6&  \Lambda_{12}-\Lambda_6-\Lambda_4+\Lambda_3+\Lambda_2-\Lambda_1 \\  \mathsf{E}_7& \Lambda_{18}-\Lambda_{9}-\Lambda_{6}+\Lambda_3+\Lambda_2-\Lambda_1  \\ \mathsf{E}_8 &  \Lambda_{30}-\Lambda_{15}-\Lambda_{10}-\Lambda_{6}+\Lambda_5+\Lambda_3+\Lambda_2-\Lambda_1\end{array}\right.
\end{equation}

We say that $\mathsf{T}$ is {\em exceptional\/} if it is in the list
of Arnold's 14 exceptional unimodal singularities
\cite{Arnold:1985kx}:
\begin{equation}
  \begin{array}{lll}
    \mathsf{T}\qquad&  (h,d_1,d_2,d_3)\qquad&f\\
    \hline
    \mathsf{U}_{12}&  (12,4,4,3)&x_1^3+x_2^3+x_3^4\\
    \mathsf{S}_{12}& (13,4,3,5)&x_1^2x_3+x_2x_3^2+x_1x_2^3\\ 
    \mathsf{Q}_{12}& (15,6,5,3)&x_1^2x_3+x_2^3+x_2^5\\
    \mathsf{S}_{11}& (16,5,4,6)&x_1^2x_3+x_2x_3^2+x_2^4\\
    \mathsf{W}_{13}&   (16,4,3,8)&x_1^4+x_1x_2^4+x_3^2\\
    \mathsf{Q}_{11}&(18,7,6,4)&x_1^2x_3+x_2^3+x_2x_3^3\\
    \mathsf{Z}_{13}&  (18,5,3,9)&x_1^3x_2+x_2^6+x_3^2\\
    \mathsf{W}_{12} &(20,5,4,10)&x_1^4+x_2^5+x_3^2\\ 
    \mathsf{Z}_{12}& (22,6,4,11)&x_1^3x_2+x_1x_2^4+x_3^2\\
    \mathsf{Q}_{10}& (24,9,8,6)&x_1^2x_3+x_2^3+x_3^4\\
    \mathsf{E}_{14}& (24,8,3,12)&x_1^3+x_2^8+x_3^2\\
    \mathsf{Z}_{11}& (30,8,6,15)&x_1^3x_2+x_2^5+x_3^2\\
    \mathsf{E}_{13}&  (30,10,4,15)&x_1^3+x_1x_2^5+x_3^2\\
    \mathsf{E}_{12}& (42,14,6,21)&x_1^3+x_2^7+x_3^2 
  \end{array}
\end{equation}
For these singularities, we have $\hat c=1+2/h$ so that we can take
$\mathpzc{D}=1$ and $\delta=-1$.  The $\chi_y$-genera with $t=y^{1/h}$
are as follows:
\begin{equation}
  \begin{array}{lll}
    \mathsf{T}\qquad & \chi_y \\
    \hline
    \mathsf{U}_{12} &1 + {t^3} + 2\,{t^4} +
    {t^6} + 2\,{t^7} + {t^8} + 2\,{t^{10}} + {t^{11}} + {t^{14}}\\
    \mathsf{S}_{12} &1 + {t^3} + {t^4} + {t^5}
    + {t^6} + {t^87} + {t^8} + {t^9} + {t^{10}} + {t^{11}} + {t^{12}} +
    {t^{15}}\\ 
    \mathsf{Q}_{12}  &1 + {t^3} + {t^5} +
    2\,{t^6} + {t^8} + {t^9} + 2\,{t^{11}} + {t^{12}} + {t^{14}} +
    {t^{17}}\\
    \mathsf{S}_{11}  &1 + {t^4} +
    {t^5} + {t^6} + {t^8} + {t^9} + {t^{10}} + {t^{12}} + {t^{13}} +
    {t^{14}} + {t^{18}}\\ 
    \mathsf{W}_{13} & 1 +
    {t^3} + {t^4} + {t^6} + {t^7} + {t^8} + {t^9} + {t^{10}} + {t^{11}}
    + {t^{12}} + {t^{14}} + {t^{15}} + {t^{18}}\\
    \mathsf{Q}_{11}  & 1 + {t^4} + {t^6} + {t^7} +
    {t^8} + {t^{10}} + {t^{12}} + {t^{13}} + {t^{14}} + {t^{16}} +
    {t^{20}}\\ 
    \mathsf{Z}_{13} & 1 + {t^3} + {t^5}
    + {t^6} + {t^8} + {t^9} + {t^{10}} + {t^{11}} + {t^{12}} + {t^{14}}
    + {t^{15}} + {t^{17}} + {t^{20}}\\
    \mathsf{W}_{12} & 1 + {t^4} + {t^5} + {t^8} +
    {t^9} + {t^{10}} + {t^{12}} + {t^{13}} + {t^{14}} + {t^{17}} +
    {t^{18}} + {t^{22}}\\
    \mathsf{Z}_{12} & 1 +
    {t^4} + {t^6} + {t^8} + {t^{10}} + 2\,{t^{12}} + {t^{14}} + {t^{16}}
    + {t^{18}} + {t^{20}} + {t^{24}}\\
    \mathsf{Q}_{10} & 1 + {t^6} + {t^8} + {t^9} +
    {t^{12}} + {t^{14}} + {t^{17}} + {t^{18}} + {t^{20}} + {t^{26}}\\
    \mathsf{E}_{14}  &1 + {t^3} + {t^6} + {t^8} + {t^9}
    + {t^{11}} + {t^{12}} + {t^{14}} + {t^{15}} + {t^{17}} + {t^{18}} +
    {t^{20}} + {t^{23}} + {t^{26}}\\
    \mathsf{Z}_{11}  &1 + {t^6} + {t^8} + {t^{12}} +
    {t^{14}} + {t^{16}} + {t^{18}} + {t^{20}} + {t^{24}} + {t^{26}} +
    {t^{32}}\\ 
    \mathsf{E}_{13}  & 1 + {t^4} +
    {t^8} + {t^{10}} + {t^{12}} + {t^{14}} + {t^{16}} + {t^{18}} +
    {t^{20}} + {t^{22}} + {t^{24}} + {t^{28}} + {t^{32}}\\
    \mathsf{E}_{12} & 1 + {t^6} + {t^{12}} + {t^{14}}
    + {t^{18}} + {t^{20}} + {t^{24}} + {t^{26}} + {t^{30}} + {t^{32}} +
    {t^{38}} + {t^{44}} 
  \end{array}
\end{equation}
The computation of $\bfL$ is again straightforward
\cite{Kawai:1992kx}:
\begin{equation}
  \begin{array}{ll}
    \mathsf{T}\qquad &\bfL\\
    \hline &\\[-4mm]
    \mathsf{U}_{12}&     \Lambda_{12}  + \Lambda_{4}- \Lambda_{3} -\Lambda_1\\
    \mathsf{S}_{12}&\Lambda_{13}-\Lambda_1\\
    \mathsf{Q}_{12}&\Lambda_{15}- \Lambda_{5} + \Lambda_{3}  -\Lambda_1\\
    \mathsf{S}_{11}&\Lambda_{16}- \Lambda_{8} + \Lambda_{4}  -\Lambda_1\\
    \mathsf{W}_{13}& \Lambda_{16} - \Lambda_{4} + \Lambda_{2}-\Lambda_1\\
    \mathsf{Q}_{11}& \Lambda_{18} - \Lambda_{9}+ \Lambda_{3} -\Lambda_1\\
    \mathsf{Z}_{13}& \Lambda_{18} - \Lambda_{6}+ \Lambda_{2} -\Lambda_1\\
    \mathsf{W}_{12}&  \Lambda_{20}- \Lambda_{10}+ \Lambda_{5}- \Lambda_{4}+ \Lambda_{2} -\Lambda_1\\
    \mathsf{Z}_{12}& \Lambda_{22} - \Lambda_{11} + \Lambda_{2} -\Lambda_1\\
    \mathsf{Q}_{10}&  \Lambda_{24}- \Lambda_{12}  - \Lambda_{8}+ \Lambda_{4}+ \Lambda_{3} -\Lambda_1\\
    \mathsf{E}_{14}& \Lambda_{24} - \Lambda_{8}- \Lambda_{6}+ \Lambda_{3}+ \Lambda_{2} -\Lambda_1\\
    \mathsf{Z}_{11}&  \Lambda_{30}- \Lambda_{15} - \Lambda_{10}+ \Lambda_{5}+ \Lambda_{2} -\Lambda_1\\
    \mathsf{E}_{13}& \Lambda_{30}- \Lambda_{15}- \Lambda_{6} + \Lambda_{3} + \Lambda_{2}  -\Lambda_1\\
    \mathsf{E}_{12}& \Lambda_{42}- \Lambda_{21}  - \Lambda_{14} + \Lambda_{7}- \Lambda_{6} + \Lambda_{3}+ \Lambda_{2}-\Lambda_1
  \end{array}
\end{equation}

To make manifest the dependence on the type under consideration, we
write $\bfL$ for type $\mathsf{T}$ as $\bfL[\mathsf{T}]$ {etc}.\ in
the following.  If $\mathsf{T}$ is simple, let its dual $\mathsf{T}^*$
be $\mathsf{T}$ itself.  If $\mathsf{T}$ is exceptional,
$\mathsf{T}^*$ is defined to be the dual of $\mathsf{T}$ in the sense
of Arnold's strange duality \cite{Arnold:1985kx}. For instance, we have
$\mathsf{S}_{11}^*=\mathsf{W}_{13}$ and
$\mathsf{E}_{12}^*=\mathsf{E}_{12}$. Note that $\mathsf{T}$ and
$\mathsf{T}^*$ share the same $h$.  We observed in \cite{Kawai:1992kx} 
that if $\mathsf{T}$ is either simple or exceptional  with
$\bfL[\mathsf{T}]=\sum_{d\mid h}u_d\Lambda_d$, then we have
\begin{equation}\label{duality}
  \bfL[\mathsf{T}^*]= -\sum_{d\vert h}u_{h/d} \Lambda_{d}.
\end{equation}
Moreover, it was anticipated there that this observation should
somehow be related to orbifoldization.  From the results in the
preceding sections it is now easy to understand why this phenomenon
must occur.  Indeed, we already saw in \cite{Kawai:1995nx} that
\begin{equation}\label{chiyduality}
  \chi_y^{\rm orb}[\mathsf{T}]=-\chi_y[\mathsf{T}^*].
\end{equation}
It is then obvious that
\begin{equation}
  \bfL^{\rm orb}[\mathsf{T}]=-\bfL[\mathsf{T}^*].
\end{equation}
On the other hand, we have
\begin{equation}
  \bfL^{\rm orb}[\mathsf{T}]=\sum_{d\vert h}u_{h/d} \Lambda_{d}
\end{equation}
since $\delta=\pm 1$   implies $h^{\rm orb}=h$ and $u^{\rm
  orb}_d=u_{h/d}$. 
So \eqref{duality} must hold by consistency.  Since $\wi^{\rm
  orb}_v[\mathsf{T}]=-\wi_v[\mathsf{T}^*]$, we see in a similar
fashion that
\begin{equation}
  \wi_v[\mathsf{T}^*]= -\sum_{d\vert h}u_{h/d} (\Delta_{d})_v.
\end{equation}
Of course, it is tempting to suppose that
\begin{equation}
  Z^{\rm orb}[\mathsf{T}](\tau,z)=-Z[\mathsf{T}^*](\tau,z)
\end{equation}
from which \eqref{chiyduality} follows.  If $\mathsf{T}$ is simple, in
which case we have $\hat c <1$, it is obvious that this holds by our
reasoning in \S \ref{sec:ASD}.

\section{The {\boldmath $\mathcal{N}=2$} minimal models and Witten's
  conjecture}\label{sec:minimal}

We begin this section by  reviewing  the elliptic genera of the
$\mathcal{N}=2$ minimal models following the description in \cite{Kawai:1994uq}. 
 For a positive integer $n$ and $m\in \ZZ_{2n}$, the theta function
$\theta_{m,n}$ is defined by
\begin{equation}
  \theta_{m,n}(\tau,u)=\sum_{j\in \ZZ+\frac{m}{2n}}\bfe{n(j^2\tau+ju)}.
\end{equation}
Fix a non-negative integer $k$ and set $h=k+2$.  The integrable
representations at level $k$ of the untwisted affine Lie algebra $\hat
A_1$ are labeled by $\ell\in \{1,2,\dots,h-1\}$ so that $(\ell-1)/2$
is the spin of the underlying representation of $A_1$.  The associated
Weyl-Kac characters are given by
\begin{equation}
  \chi_\ell(\tau,u)=\frac{\theta_{\ell,h}(\tau,u)-\theta_{-\ell,h}(\tau,u)}{\theta_{1,2}(\tau,u)-\theta_{-1,2}(\tau,u)}.
\end{equation}
Then the string functions $c^\ell_m(\tau)$ are defined through
\cite{Kac:1984fk}
\begin{equation}
  \chi_\ell(\tau,u)=\sum_{m \in \ZZ_{2k}}c^\ell_m(\tau)\theta_{m,k}(\tau,u).
\end{equation}
Note that $c^\ell_m(\tau)$ vanishes unless $\ell-1\equiv m\pmod{2}$.

The branching relation introduced by Gepner \cite{Gepner:1988yq} for
the $\mathcal{N}=2$ minimal models has to be slightly extended
\cite{Kawai:1994uq} so that the dependence on the variable $z$ which 
measures the $U(1)$ charge is manifest:
\begin{equation}
  \chi_\ell(\tau,u)\theta_{a,2}(\tau,u-z)=\sum_{m\in \ZZ_{2h}} \chi^{\ell,a}_m(\tau,z)\theta_{m,h}(\tau,u-\frac{2}{h}z)
\end{equation}
where $a \in \ZZ_4$.
This can be solved by the multiplication formula of theta functions  as
\begin{equation}
  \chi^{\ell,a}_m(\tau,z)=\sum\limits_{j\in \mathbb{ Z}}
  c^ \ell_{m-a+4j}(\tau)
  q^{\frac{h}{2k}\left(\frac{m}{h}-\frac{a}{2}+2j\right)^2}
  y^{\frac{m}{h}-\frac{a}{2}+2j}.
\end{equation}
We then define
\begin{equation}
  \II^ \ell_m(\tau,z)=\chi_m^{ \ell,1}(\tau,z)-\chi_m^{ \ell,-1}(\tau,z).
\end{equation}
Note that $\II^ \ell_m(\tau,z)$ vanishes unless $\ell \equiv
m\pmod{2}$.  Actually, $\II^ \ell_m(\tau,z)$ coincides with the
(twisted) character in the representation theory of the Ramond
$\mathcal{N}=2$ SCA \cite{Dobrev:1987dz,Matsuo:1987fj,Kiritsis:1988hi}. Though this is expected from our definition, we
will provide an explicit proof shortly.  It is easy to check $\II^
\ell_{-m}(\tau,z)=-\II^ \ell_m(\tau,-z)=-\II^ {h-\ell}_{h-m}(\tau,z)$
as well as
\begin{equation}
  \II^ \ell_m(\tau,0)=\delta_{m,\ell}^{(2h)}-\delta_{m,-\ell}^{(2h)}.
\end{equation}

As is well-known, the modular invariants of the $\hat A_1$
Wess-Zumino-Witten models at level $k$ fall into the
$\mathsf{A}\mathsf{D}\mathsf{E}$ pattern
\cite{Cappelli:1987uq,Cappelli:1987qy,Kato:1987vn,Gannon:2000kx}:
\begin{equation}
  \begin{array}{ll}
    \mathsf{T}&\sum_{\ell,\ell'}\mathscr{N}_{\ell,\ell'}\chi_\ell\overline{\chi_{\ell'}}\\[1mm]
    \hline\\[-3mm]
    \mathsf{A}_{h-1}&\sum_{\ell=1}^{h-1}\abs{\chi_\ell}^2\\[2mm]
    \mathsf{D}_{\frac{h}{2}+1}\ (2 \nmid \frac{h}{2})\qquad &\sum_{i=1}^{\frac{h-2}{4}}\abs{\chi_{2i-1}+\chi_{h-(2i-1)}}^2+2\abs{\chi_{\frac{h}{2}}}^2\\[2mm]
    \mathsf{D}_{\frac{h}{2}+1}\ (2\mid \frac{h}{2})&\sum_{i=1}^{\frac{h}{2}}\abs{\chi_{2i-1}}^2+   \abs{\chi_{\frac{h}{2}}}^2    +\sum_{i=1}^{\frac{h}{4}-1}(\chi_{2i}\overline{\chi_{h-2i}}+\chi_{h-2i}\overline{\chi_{2i}} )\\[2mm]
    \mathsf{E}_{6}&\abs{\chi_1+\chi_7}^2+\abs{\chi_4+\chi_8}^2+\abs{\chi_5+\chi_{11}}^2
    \\[2mm]
    \mathsf{E}_{7}&\abs{\chi_1+\chi_{17}}^2+\abs{\chi_5+\chi_{13}}^2+\abs{\chi_7+\chi_{11}}^2\\[2mm]
    &+\abs{\chi_9}^2+
    (\chi_3+\chi_{15})\overline{\chi_9}+\chi_9\overline{(\chi_3+\chi_{15})}
    \\[2mm]
    \mathsf{E}_{8}&\abs{\chi_1+\chi_{11}+\chi_{19}+\chi_{29}}^2+\abs{\chi_7+\chi_{13}+\chi_{17}+\chi_{23}}^2
    \\[2mm]
  \end{array}\label{A1-modular}
\end{equation}
where we assume as before that $h\ge 2$ for $\mathsf{A}_{h-1}$ and
$h\ge 6$, $h\in 2\ZZ$ for $\mathsf{D}_{\frac{h}{2}+1}$.  Note the
symmetry $\mathscr{N}_{ h-\ell, h-\ell'}=\mathscr{N}_{\ell,\ell'}$.  After picking up a
modular invariant from this list the elliptic genus is given by
\begin{equation}
  Z(\tau,z)=\sum_{ \ell, \ell'=1}^{h-1}\mathscr{N}_{ \ell, \ell'}\, \II^\ell_{\ell'}(\tau,z).
\end{equation}
That this expression satisfies the desired functional properties
\eqref{eg1} and \eqref{eg2} with $\hat c =1-2/h$ was confirmed in
\cite{Kawai:1994uq}.

It is straightforward to find the expressions for the twisted elliptic
genera:
\begin{equation}\label{twistedminimal}
  Z_{(\twista,\twistb)}(\tau,z)=\sum_{ \ell, \ell'=1}^{h-1}\mathscr{N}_{\ell,\ell'}\,\zeta_h^{(\ell'-\twista)\twistb} \,\II^\ell_{\ell'-2r}(\tau,z),\quad (\twista,\twistb)\in (\ZZ_h)^2.
\end{equation}
From this it is easy to show that $Z^{\rm orb}(\tau,z)=-Z(\tau,z)$ as
already claimed in \cite{Kawai:1994uq}.

Now we are in a position to prove 
Witten's conjecture \cite{Witten:1994fj} which amounts to say that the
elliptic genera of the $\mathcal{N}=2$ minimal model and the LG model
both labeled by the same $\mathsf{ADE}$ type must coincide. 
We should first confirm that the elliptic
genera of both theories are holomorphic with respect to $z$.  For the
LG model this has already been seen.  As for the minimal model this
can be confirmed for instance from Proposition~\ref{productformula}
below.  That Assumption~\ref{tWgenus} is satisfied by both theories is
also apparent.  Since the two theories have the same central charge
$\hat c =1-2/h$ and their elliptic genera satisfy the same functional
equations, we can apply Proposition \ref{egcriteria}. 
Since $\hat c<1$, it suffices to show that both theories have either
 the same twisted Witten indices or the same $\chi_y$-genus.  That both
have the same $\chi_y$-genus is easy to show and in fact well-known:
$\chi_y=\sum_{\ell} \mathscr{N}_{\ell,\ell}y^{(\ell-1)/h}=\sum_i y^{(m_i-1)/h}$.
So we are done.

It is also not too difficult to show directly that both have the same
twisted Witten indices.  To explain this, it is convenient to extend
the definition of $\mathscr{N}_{\ell,\ell'}$ for all $\ell,\ell'\in \ZZ$ by
\begin{equation}\label{Nextension}
  \mathscr{N}_{\epsilon\ell,\epsilon'\ell'}=\epsilon\epsilon'\mathscr{N}_{\ell,\ell'}\quad (\epsilon,\epsilon'=\pm 1),\qquad \mathscr{N}_{\ell+2p,\ell'+2p'}=\mathscr{N}_{\ell,\ell'}\quad (p,p'\in h\ZZ).
\end{equation}
We then recall\footnote{Historically speaking, this is the way the
  $\mathscr{N}_{\ell,\ell'}$ were first found in \cite{Cappelli:1987uq} and
  the $(\Omega_d)_{\ell,\ell'}$ are associated with the modular
  invariants for theta functions.}  that
\begin{equation}\label{Omegadecomp}
  \mathscr{N}_{\ell,\ell'}=\sum_{d\mid h}u_d(\Omega_d)_{\ell,\ell'}
\end{equation}
where
\begin{equation}
  (\Omega_d)_{\ell,\ell'}=\begin{cases}1&\text{if $\frac{h}{d}\mid\frac{\ell+\ell'}{2}$ and $d\mid \frac{\ell-\ell'}{2}$}\\
    0&\text{otherwise}
  \end{cases}
\end{equation}
and the $u_d$ are precisely the ones given in \eqref{ADEexp}.  Note
that this expression is consistent with \eqref{Nextension} since
$u_d=-u_{h/d}$ and
$(\Omega_d)_{\ell,-\ell'}=(\Omega_{h/d})_{\ell,\ell'}$.  It follows
from \eqref{twistedminimal} that
\begin{equation}
  \begin{split}
    \wi_{(\twista,\twistb)}&=\frac{1}{2}\sum_{(\ell,\ell')\in (\ZZ_{2h})^2}\mathscr{N}_{\ell,\ell'}\zeta_h^{(\ell'-\twista)\twistb}\delta^{(2h)}_{\ell,\ell'-2\twista}\\
    &=\frac{1}{2}\sum_{\substack{ (\ell,\ell')\in (\ZZ_{2h})^2\\
        \ell\equiv\ell' \!\!
        \!\pmod{2}}}\mathscr{N}_{\ell,\ell'}\zeta_h^{(\frac{\ell+\ell'}{2})\twistb}\delta^{(h)}_{\frac{\ell-\ell'}{2},-\twista}.
  \end{split}
\end{equation}
Therefore, all we have to show is
\begin{equation}
  \frac{1}{2}\sum_{\substack{(\ell,\ell')\in \ZZ_{2h}^2\\ \ell\equiv\ell' \!\!\!\!\pmod{2}}}(\Omega_d)_{\ell,\ell'}\zeta_h^{(\frac{\ell+\ell'}{2})\twistb}\delta^{(h)}_{\frac{\ell-\ell'}{2},-\twista}=(\Delta_d)_{(\twista,\twistb)},
\end{equation}
but this can be readily checked.

In the remainder of this section  we  reformulate Witten's conjecture which we have just
proved as mathematical identities generalizing the QPI.  For this
purpose, it is convenient to borrow several notations from $q$-analysis.
If $q$ and $x$ are complex numbers with $\abs{q}<1$, set
\begin{equation}
  (x;q)_\infty=\prod_{n=0}^\infty(1-xq^n),\qquad (q)_\infty=(q;q)_\infty\\
\end{equation}
and if in addition $x\ne 0$, put
\begin{equation}
  \theta(x,q)=\frac{(x;q)_\infty(q/x;q)_\infty}{(q)_\infty^2}\,.
\end{equation}
By Jacobi's triple product identity, we have
\begin{equation}\label{Jacobitriple}
  (x;q)_\infty(q/x;q)_\infty(q;q)_\infty=(q)_\infty^3\theta(x,q)=\sum_{n\in \ZZ}(-1)^nq^{\binom{n}{2}}x^n
\end{equation}
where $\binom{n}{2}=\frac{n(n-1)}{2}$.  The following identities
easily follow from the definitions:
\begin{align}\label{eq:theta}
  \theta(q/x,q)&=\theta(x,q)\,,\\
  \theta(x^{-1},q)&=-x^{-1}\theta(x,q)\,,\\
  \theta(q^i x,q)&=(-1)^iq^{-\binom{i}{2}}x^{-i}\theta(x,q) \quad
  ( i\in \ZZ).
\end{align}
Moreover, for a positive integer $h$, we have
\begin{equation}
  \theta(x,q)=\frac{(q^h)_\infty^{2h}}{(q)_\infty^2}\prod_{s=0}^{h-1}  \theta(q^s x,q^h).
\end{equation}
 By comparing the behaviors of both sides
as $x\to 1$, we see that 
\begin{equation}
  1=\frac{(q^h)_\infty^{2h}}{(q)_\infty^2}\prod_{s=1}^{h-1}  \theta(q^s,q^h).
\end{equation}
It thus follows that
\begin{equation}\label{theta-id}
  \frac{\theta(x,q)}{\theta(x,q^h)}=\prod_{s=1}^{h-1}  \frac{\theta(q^s x,q^h)}{\theta(q^s,q^h)}.
\end{equation}

We also quote the following results of Kronecker \cite{Weil:1976lr}:
\begin{prop}
  For $0<\abs{q}<\abs{x}<1$ and $y$ neither $0$ nor an integral power
  of $q$,
  \begin{equation}\label{eq:Kronecker-1}
    \sum_{r \in \ZZ}\frac{x^r}{1-q^r y}=\frac{\theta(xy,q)}{\theta(x,q)\theta(y,q)}\,.
  \end{equation}
  For $\abs{q}<\abs{x}<1$ and $\abs{q}<\abs{y}<1$,
  \begin{equation}\label{eq:Kronecker-2}
    \left(\sum_{r,s\ge 0}-\sum_{r,s<0}\right)q^{rs}x^ry^s=\frac{\theta(xy,q)}{\theta(x,q)\theta(y,q)}\,.
  \end{equation}
\end{prop}

While our definition of $\II^\ell_m$ has the advantage of an easy
access to its modular properties, its description in terms of string
functions is rather awkward.  Fortunately, there is an alternative
expression:
\begin{prop}
  We have
  \begin{equation}\label{productformula}
    \II^\ell_m(\tau,z)= q^{\frac{ \ell^2-m^2}{4h}}     y^{\frac{m-1}{h}-\frac{\hat c}{2}}
    \frac{\theta(y,q)\theta(q^ \ell,q^h)} {
      \theta(q^{\frac{ \ell+m}{2}}y^{-1},q^h)
      \theta(q^{\frac{ \ell-m}{2}}y,q^h)}
  \end{equation}
  where $\ell \equiv m \pmod{2}$ should be understood.
\end{prop}
\begin{proof}
  Consider first the case $h=2$.  If $(\ell,m)=(1, 1)$, then the RHS
  of \eqref{productformula} reads
  \begin{equation}
    \frac{\theta(y,q)\theta(q,q^2)} {
      \theta(qy^{-1},q^2)
      \theta(y,q^2)}=\frac{\theta(y,q)}{\theta(y,q^2)}\frac{\theta(q,q^2)}{\theta(qy,q^2)}=1
  \end{equation}
  where we used \eqref{theta-id} in the last equality. On the other
  hand $\II^1_1(\tau,z)=1$, so the assertion holds.  The situation for
  $(\ell,m)=(1,-1)$ is similar.

  Suppose next that $h>2$. The string function $c^\ell_m$ is well-known
  to be expressed as \cite{Kac:1984fk}
  \begin{equation}
    c^ \ell_m(\tau)=\frac{q^{d_{ \ell,m}}}{(q)_\infty^3}\left(\sum_{r,s\ge 0}-\sum_{r,s<0}\right)(-1)^{r+s}q^{\binom{r+s+1}{2}+krs+\frac{ \ell-1+m}{2}r+\frac{ \ell-1-m}{2}s}
  \end{equation}
  where $ d_{ \ell,m}=\frac{ \ell^2}{4h}-\frac{m^2}{4k}-\frac{1}{8} $.
  Note that some terms in the sum are spurious due to the identity
  $\sum_{n\in \ZZ}(-1)^nq^{\binom{n}{2}+na}=0$ $(a\in \ZZ)$ which
  follows from \eqref{Jacobitriple} and $\theta(q^a,q)=0$.  We thus
  obtain from the definition of $\II_m^\ell(\tau,z)$ that
  \begin{equation}\label{tmp-I}
    \begin{split}
      &(q^{\frac{ \ell^2-m^2}{4h}} y^{\frac{m-1}{h}-\frac{\hat c}{2}} )^{-1}\,\II^\ell_m(\tau,z)\\
      &\qquad =\frac{1}{(q)_\infty^3} \sum_{n\in \ZZ}
      \left(\sum_{r,s\ge 0}-\sum_{r,s<0}\right)
      (-1)^{n+r+s}q^{\binom{n+r-s}{2}+hrs+\frac{ \ell+m}{2}r+\frac{
          \ell-m}{2}s} y^n.
    \end{split}
  \end{equation}
  On the other hand, we see from \eqref{Jacobitriple} and
  \eqref{eq:Kronecker-2} that
  \begin{equation}
    \begin{split}
      \theta(y,q) & \frac{ \theta(q^ \ell,q^h)} { \theta(q^{\frac{
            \ell+m}{2}}y^{-1},q^h)
        \theta(q^{\frac{ \ell-m}{2}}y,q^h)}\\
      &=\inv{(q)_\infty^3} \sum_{i\in \ZZ}(-1)^iy^iq^{\binom{i}{2}}
      \cdot
      \left(\sum_{r,s\ge 0}-\sum_{r,s<0}\right) q^{hrs}\left(q^{\frac{ \ell+m}{2}}y^{-1}\right)^r\left(q^{\frac{ \ell-m}{2}}y\right)^s\\
      &=\inv{(q)_\infty^3}\sum_{i\in \ZZ} \left(\sum_{r,s\ge
          0}-\sum_{r,s<0}\right) (-1)^iq^{\binom{i}{2}+hrs+\frac{
          \ell+m}{2}r+\frac{ \ell-m}{2}s}y^{i-r+s}.
    \end{split}
  \end{equation}
  By replacing $i$ by $n+r-s$ this is seen to coincide with
  \eqref{tmp-I}.
\end{proof}
\begin{rem} As promised, the RHS of \eqref{productformula} is in fact
  the product expression found in \cite{Matsuo:1987fj} for the (twisted)
  character of the Ramond $\mathcal{N}$=2 SCA.
\end{rem}

With this provision, we can restate our result as
\begin{thm} Suppose that $\mathsf{T}$ is of $\mathsf{ADE}$ type with $(h,d_1,d_2,d_3)$ 
  and $\mathscr{N}_{\ell,\ell'}$ given respectively by
  \eqref{ADEtypes} and \eqref{A1-modular}.  Then,
  \begin{equation}\label{eq:GQPI}
    \prod_{i=1}^{3}\frac{\theta(x^{h-d_i},q)}{\theta(x^{d_i},q)}=\sum_{\ell,\ell'=1}^{h-1}\mathscr{N}_{\ell,\ell'}\,\JJ^\ell_{\ell'}(x,q)
  \end{equation}
  where
  \begin{equation}
    \JJ^\ell_{\ell'}(x,q)=q^{\frac{ \ell^2-\ell' {}^2}{4h}} x^{
      \ell'-1}     
    \frac{\theta(x^h,q)\theta(q^ \ell,q^h)} {
      \theta(q^{\frac{ \ell+  \ell'}{2}}x^{-h},q^h)
      \theta(q^{\frac{ \ell-  \ell'}{2}}x^h,q^h)}.
  \end{equation}
\end{thm}
\begin{rem}
  Note that the LHS of \eqref{eq:GQPI}  is a redundant expression and can actually be simplified as
  $\displaystyle
  \prod_{i=1}^{n'}\frac{\theta(x^{h-d_i},q)}{\theta(x^{d_i},q)}$ where
  $n'=1$ if $\mathsf{T}=\mathsf{A}_{h-1}$ and $n'=2$ if otherwise.
\end{rem}  
This theorem can be interpreted as $\mathsf{A}\mathsf{D}\mathsf{E}$
generalizations of the QPI for the following reason.  (See
\cite{Cooper:2006qy} for a comprehensive survey on the history, many
available proofs and some generalizations of the QPI.)  When
$\mathsf{T}=\mathsf{A}_{h-1}$, the explicit form of \eqref{eq:GQPI} is
\begin{equation}\label{eq:A-series}
  \frac{\theta(x^{h-1},q)}{\theta(x,q)}=\sum_{ \ell=1}^{h-1}x^{ \ell-1}
  \frac{\theta(x^h,q)\theta(q^ \ell,q^h)}{\theta(q^ \ell x^{-h},q^h)\theta(x^h,q^h)}\,.
\end{equation}
As before, it is easy to confirm the case $h=2$. The first
non-trivial case $h=3$ turns out to be just the QPI (expressed in one
among many other possible ways \cite{Cooper:2006qy}).  By using the
$h$ term relation of  the {Weierstra\ss}  $\sigma$-function, Bailey \cite{Bailey:1953lr} was
able to generalize the QPI as
\begin{equation}\label{eq:Bailey}
  \frac{\theta(x^{h-1},q)}{\theta(x,q)}=\prod_{s=1}^{h-1}\frac{ \theta(q^s x^h,q^h)}{ \theta(q^s,q^h)}\sum_{ \ell=1}^{h-1}x^{ \ell-1}\frac{\theta(q^ \ell,q^h)}{\theta(q^ \ell x^{-h},q^h)}\,.
\end{equation}
With the aid of \eqref{theta-id}, it is easy to see the equivalence
between \eqref{eq:A-series} and \eqref{eq:Bailey}.

It should be mentioned that we can prove \eqref{eq:A-series} by using
\eqref{eq:Kronecker-1} as well.  Indeed, if one notices
$\theta(q^h,q^h)=0$, the RHS of \eqref{eq:A-series} can be rewritten
as
\begin{equation}
  \begin{split}
    \theta(x^h,q) \sum_{ \ell=1}^{h}x^{ \ell-1}\frac{\theta(q^
      \ell,q^h)}{\theta(q^ \ell x^{-h},q^h)\theta(x^h,q^h)} &=
    \theta(x^h,q)
    \sum_{ \ell=1}^hx^{ \ell-1}\sum_{r\in \ZZ}\frac{(q^ \ell x^{-h})^r}{1-q^{hr}x^h}\\
    &=\theta(x^h,q) \sum_{r\in\ZZ}\frac{(q/x^h)^r}{1-q^r x}\,.
  \end{split}
\end{equation}
On the other hand, we have
\begin{equation}\label{ex:Baileytmp}
  \frac{\theta(x^{h-1},q)}{\theta(x,q)}=\theta(x^h,q)
  \frac{\theta(q/x^{h-1},q)}{\theta(q/x^h,q)\theta(x,q)}
  =\theta(x^h,q)
  \sum_{r\in\ZZ}\frac{(q/x^h)^r}{1-q^r x}\,.
\end{equation}

In view of the fact that the QPI admits a variety of derivations
\cite{Cooper:2006qy} it might be interesting to pursue alternative
proofs of the theorem.

\acknowledgments The author is supported by KAKENHI (19540024).


\providecommand{\href}[2]{#2}\begingroup\raggedright\endgroup


\begin{thebibliography}{10}

\bibitem{Kawai:1994uq}
T.~Kawai, Y.~Yamada, and S.-K. Yang, {\it Elliptic genera and {$N=2$}
  superconformal field theory},  {\em Nuclear Phys. B} {\bf 414} (1994),
  no.~1-2 191--212, [\href{http://arxiv.org/abs/arXiv:hep-th/9306096}{{\tt
  arXiv:hep-th/9306096}}].

\bibitem{Schwimmer:1987jk}
A.~Schwimmer and N.~Seiberg, {\it {Comments on the N=2, N=3, N=4 Superconformal
  Algebras in Two-Dimensions}},  {\em Phys. Lett.} {\bf B184} (1987) 191--196.

\bibitem{Eichler:1985kx}
M.~Eichler and D.~Zagier, {\em The theory of {J}acobi forms}, vol.~55 of {\em
  Progress in Mathematics}.
\newblock Birkh\"auser Boston Inc., Boston, MA, 1985.

\bibitem{Atkin:1954lr}
A.~O.~L. Atkin and P.~Swinnerton-Dyer, {\it Some properties of partitions},
  {\em Proc. London Math. Soc. (3)} {\bf 4} (1954) 84--106.

\bibitem{Kawai:1992kx}
T.~Kawai, {\it Fine structure of the {MVW} correspondence},  {\em Internat. J.
  Modern Phys. A} {\bf 7} (1992), no.~11 2371--2415.

\bibitem{Witten:1994fj}
E.~Witten, {\it On the {L}andau-{G}inzburg description of {$N=2$} minimal
  models},  {\em Internat. J. Modern Phys. A} {\bf 9} (1994), no.~27
  4783--4800, [\href{http://arxiv.org/abs/arXiv:hep-th/9304026}{{\tt
  arXiv:hep-th/9304026}}].

\bibitem{Bailey:1953lr}
W.~N. Bailey, {\it An expression for {$\vartheta\sb 1(nz)/\vartheta\sb 1(z)$}},
   {\em Proc. Amer. Math. Soc.} {\bf 4} (1953) 569--572.

\bibitem{Lerche:1989cy}
W.~Lerche, C.~Vafa, and N.~P. Warner, {\it Chiral rings in {$N=2$}
  superconformal theories},  {\em Nuclear Phys. B} {\bf 324} (1989), no.~2
  427--474.

\bibitem{Fre:1991yq}
P.~Fr{\'e}, F.~Gliozzi, M.~R. Monteiro, and A.~Piras, {\it A moduli-dependent
  {L}agrangian for {$(2,2)$} theories on {C}alabi-{Y}au {$n$}-folds},  {\em
  Classical Quantum Gravity} {\bf 8} (1991), no.~8 1455--1470.

\bibitem{Fre:1992kx}
P.~Fr{\'e}, L.~Girardello, A.~Lerda, and P.~Soriani, {\it Topological
  first-order systems with {L}andau-{G}inzburg interactions},  {\em Nuclear
  Phys. B} {\bf 387} (1992), no.~2 333--372.

\bibitem{Milnor:1968lr}
J.~Milnor, {\em Singular points of complex hypersurfaces}.
\newblock Annals of Mathematics Studies, No. 61. Princeton University Press,
  Princeton, N.J., 1968.

\bibitem{Steenbrink:1977fk}
J.~Steenbrink, {\it Intersection form for quasi-homogeneous singularities},
  {\em Compositio Math.} {\bf 34} (1977), no.~2 211--223.

\bibitem{Arnold:1988qy}
V.~I. Arnold, S.~M. Gusein-Zade, and A.~N. Varchenko, {\em Singularities of
  differentiable maps. {V}ol. {II}}, vol.~83 of {\em Monographs in
  Mathematics}.
\newblock Birkh\"auser Boston Inc., Boston, MA, 1988.

\bibitem{Orlik:1977lr}
P.~Orlik and L.~Solomon, {\it Singularities. {II}. {A}utomorphisms of forms},
  {\em Math. Ann.} {\bf 231} (1977/78), no.~3 229--240.

\bibitem{Milnor:1970fk}
J.~Milnor and P.~Orlik, {\it Isolated singularities defined by weighted
  homogeneous polynomials},  {\em Topology} {\bf 9} (1970) 385--393.

\bibitem{ACampo:1975fj}
N.~A'Campo, {\it La fonction z\^eta d'une monodromie},  {\em Comment. Math.
  Helv.} {\bf 50} (1975) 233--248.

\bibitem{Kawai:1995nx}
T.~Kawai and S.-K. Yang, {\it Duality of orbifoldized elliptic genera},  {\em
  Progr. Theoret. Phys. Suppl.} (1995), no.~118 277--297,
  [\href{http://arxiv.org/abs/arXiv:hep-th/9408121}{{\tt
  arXiv:hep-th/9408121}}]. Quantum field theory, integrable models and beyond
  (Kyoto, 1994).

\bibitem{Vafa:1989gd}
C.~Vafa, {\it String vacua and orbifoldized {LG} models},  {\em Modern Phys.
  Lett. A} {\bf 4} (1989), no.~12 1169--1185.

\bibitem{Arnold:1985kx}
V.~I. Arnold, S.~M. Gusein-Zade, and A.~N. Varchenko, {\em Singularities of
  differentiable maps. {V}ol. {I}}, vol.~82 of {\em Monographs in Mathematics}.
\newblock Birkh\"auser Boston Inc., Boston, MA, 1985.

\bibitem{Kac:1984fk}
V.~G. Kac and D.~H. Peterson, {\it Infinite-dimensional {L}ie algebras, theta
  functions and modular forms},  {\em Adv. in Math.} {\bf 53} (1984), no.~2
  125--264.

\bibitem{Gepner:1988yq}
D.~Gepner, {\it Space-time supersymmetry in compactified string theory and
  superconformal models},  {\em Nuclear Phys. B} {\bf 296} (1988), no.~4
  757--778.

\bibitem{Dobrev:1987dz}
V.~K. Dobrev, {\it Characters of the unitarizable highest weight modules over
  the {$N=2$} superconformal algebras},  {\em Phys. Lett. B} {\bf 186} (1987),
  no.~1 43--51.

\bibitem{Matsuo:1987fj}
Y.~Matsuo, {\it Character formula of {$c<1$} unitary representation of {$N=2$}
  superconformal algebra},  {\em Progr. Theoret. Phys.} {\bf 77} (1987), no.~4
  793--797.

\bibitem{Kiritsis:1988hi}
E.~B. Kiritsis, {\it Character formulae and the structure of the
  representations of the {$N=1,\;N=2$} superconformal algebras},  {\em
  Internat. J. Modern Phys. A} {\bf 3} (1988), no.~8 1871--1906.

\bibitem{Cappelli:1987uq}
A.~Cappelli, C.~Itzykson, and J.-B. Zuber, {\it Modular invariant partition
  functions in two dimensions},  {\em Nuclear Phys. B} {\bf 280} (1987), no.~3
  445--465.

\bibitem{Cappelli:1987qy}
A.~Cappelli, C.~Itzykson, and J.-B. Zuber, {\it The {${\rm A}$}-{${\rm
  D}$}-{${\rm E}$} classification of minimal and {$A\sp {(1)}\sb 1$} conformal
  invariant theories},  {\em Comm. Math. Phys.} {\bf 113} (1987), no.~1 1--26.

\bibitem{Kato:1987vn}
A.~Kato, {\it Classification of modular invariant partition functions in two
  dimensions},  {\em Modern Phys. Lett. A} {\bf 2} (1987), no.~8 585--600.

\bibitem{Gannon:2000kx}
T.~Gannon, {\it The {C}appelli-{I}tzykson-{Z}uber {A}-{D}-{E} classification},
  {\em Rev. Math. Phys.} {\bf 12} (2000), no.~5 739--748,
  [\href{http://arxiv.org/abs/arXiv:math/9902064}{{\tt arXiv:math/9902064}}].

\bibitem{Weil:1976lr}
A.~Weil, {\em Elliptic functions according to {E}isenstein and {K}ronecker}.
\newblock Springer-Verlag, Berlin, 1976.
\newblock Ergebnisse der Mathematik und ihrer Grenzgebiete, Band 88.

\bibitem{Cooper:2006qy}
S.~Cooper, {\it The quintuple product identity},  {\em Int. J. Number Theory}
  {\bf 2} (2006), no.~1 115--161.

\end{thebibliography}
\end{document}